\documentclass[a4paper,onecolumn,superscriptaddress,unpublished]{quantumarticle}

\pdfoutput=1

\usepackage[T1]{fontenc} %256 bit font encoding
\usepackage[english]{babel} %english language

\usepackage[numbers,sort&compress]{natbib}

\providecommand{\pra}{Phys.\ Rev.\ A}

\providecommand{\prl}{Phys.\ Rev.\ Lett.}

\usepackage[table,usenames,dvipsnames]{xcolor}
\usepackage[colorlinks=true,
hypertexnames=false
]{hyperref}
\PassOptionsToPackage{linktocpage}{hyperref} %link numbers in TOC for arxiv
\usepackage{amsmath,amssymb,mathtools,amsthm,dsfont}
\usepackage{easybmat} %block matrices
\usepackage[capitalise]{cleveref} %easy referencing of different objects

\usepackage{booktabs}
\usepackage{multirow}

\usepackage{paralist}
\usepackage{graphicx}
\usepackage{subcaption}
\usepackage[numbers,sort&compress]{natbib} 

\usepackage{xparse}
\usepackage[printonlyused,withpage,nolist]{acronym}

\usepackage{tikz}

\newtheorem{theorem}{Theorem}

\newtheorem{lemma}[theorem]{Lemma}
\newtheorem{proposition}[theorem]{Proposition}

\theoremstyle{definition}
\renewcommand{\i}{\ensuremath\mathrm{i}}

\DeclareMathOperator{\LandauO}{\mathrm{O}}

\DeclareMathOperator{\Tr}{Tr}

\newcommand{\fro}{\mathrm{F}}

\renewcommand{\L}{\operatorname{L}} %linear operators

\DeclareMathOperator{\U}{U}%unitary group U(n)
\DeclareMathOperator{\wt}{wt}

\newcommand{\gdim}{d} % global dimension n-qubits system
\newcommand{\ldim}{p} % local dimension qupit systems
\newcommand{\flip}{F}
\newcommand{\frameop}{S} % frame operator
\DeclareMathOperator{\frameopPeriodic}{\Sigma_{\mathrm{pb}}} %frame operator periodic boundary conditions
\DeclareMathOperator{\frameopOpen}{\Sigma_{\mathrm{ob}}} %frame operator open boundary conditions
\DeclareMathOperator{\suppBW}{\mathrm{supp}_{\mathrm{BW}}}
\DeclareMathOperator{\suppLC}{\mathrm{supp}_{\mathrm{LC}}}
\DeclareMathOperator{\partBW}{\mathrm{part}_{\mathrm{BW}}}
\DeclareMathOperator{\varBW}{\sigma_{\mathrm{BW}}^2}
\DeclareMathOperator{\varLC}{\sigma_{\mathrm{LC}}^2}
\DeclareMathOperator{\stdevBW}{\sigma_{\mathrm{BW}}}
\DeclareMathOperator{\stdevLC}{\sigma_{\mathrm{LC}}}
\DeclareMathOperator{\tper}{t_{\mathrm{pb}}}
\DeclareMathOperator{\topen}{t_{\mathrm{ob}}}
\DeclareMathOperator{\vfull}{v_{\mathrm{full}}} %Z operator supported on all qubits
\DeclareMathOperator{\vhalf}{v_{\mathrm{half}}} %Z operator supported on half alternating qubits
\DeclareMathOperator{\vthres}{v_{\mathrm{thres}}} %Z operator supported on treshold
\DeclareMathOperator{\stab}{\mathrm{STAB}}

\DeclareMathOperator{\Cl}{Cl}
\newcommand{\Sp}[2]{\mathrm{Sp}_{#1}(#2)}

\newcommand{\CC}{\mathbb{C}}

\newcommand{\ZZ}{\mathbb{Z}}
\newcommand{\NN}{\mathbb{N}}
\newcommand{\1}{\mathds{1}}
\newcommand{\EE}{\mathbb{E}}

\newcommand{\FF}{\mathbb{F}}

\newcommand{\C}{\mathbb{C}}

\newcommand{\fr}[1]{\mathfrak{#1}}

\newcommand{\frS}{\fr{S}}
\newcommand{\Sym}{\frS}% symmetric group

\newcommand{\argdot}{{\,\cdot\,}}

\newcommand{\ad}{\dagger}

\DeclarePairedDelimiterX{\abs}[1]{\lvert}{\rvert}{%
  \ifblank{#1}{\,\cdot\,}{#1}
}   % absolute value

\DeclarePairedDelimiterX\norm[1]\lVert\rVert{%
  \ifblank{#1}{\,\cdot\,}{#1}
}   % norm

\DeclarePairedDelimiterX{\iiiNorm}[1]{\lvert}{\rvert}{%
  \delimsize\lvert\delimsize\lvert#1\delimsize\rvert\delimsize\rvert%
}

\DeclarePairedDelimiterXPP\snorm[1]{}\lVert\rVert{_\infty}{\ifblank{#1}{\,\cdot\,}{#1}}   %spectral norm  =  (2->2)-norm

\DeclarePairedDelimiterXPP\twonorm[1]{}\lVert\rVert{_2}{\ifblank{#1}{\,\cdot\,}{#1}}   % 2-norm

\DeclarePairedDelimiterXPP\trnorm[1]{}\lVert\rVert{_1}{\ifblank{#1}{\,\cdot\,}{#1}}   % trace norm

\DeclarePairedDelimiterXPP\fnorm[1]{}\lVert\rVert{_{\fro}}{\ifblank{#1}{\,\cdot\,}{#1}}   % Fro-norm

\DeclarePairedDelimiterXPP\dnorm[1]{}\lVert\rVert{_\diamond}{\ifblank{#1}{\,\cdot\,}{#1}}   % diamond norm

\DeclarePairedDelimiterXPP\cbnorm[1]{}\lVert\rVert{_\mathrm{cb}}{\ifblank{#1}{\,\cdot\,}{#1}}   % CB-norm
\DeclarePairedDelimiterXPP\onenorm[1]{}\lVert\rVert{_{1\rightarrow 1}}{\ifblank{#1}{\,\cdot\,}{#1}}   % (1->1)-norm
\DeclarePairedDelimiterXPP\ddnorm[1]{}\lVert\rVert{_{\diamond\rightarrow \diamond}}{\ifblank{#1}{\,\cdot\,}{#1}}   % (\diamond->\diamond)-norm
\DeclarePairedDelimiterXPP\ssnorm[1]{}\lVert\rVert{_{\infty\rightarrow\infty}}{\ifblank{#1}{\,\cdot\,}{#1}}   % (\infty->\infty)-norm

\DeclarePairedDelimiterX\Set[1]\{\}{%
  
  #1
}

\DeclarePairedDelimiterX\innerp[2]{\langle}{\rangle}{%
  \ifblank{#1}{\,\cdot\,}{#1} , \ifblank{#2}{\,\cdot\,}{#2}%
}

\DeclarePairedDelimiter{\ket}{\vert}{\rangle}

\DeclarePairedDelimiterX\braket[2]{\langle}{\rangle}%
  {#1\kern0.15ex\delimsize\vert\kern0.15ex\mathopen{}#2}

\DeclarePairedDelimiterX\ketbra[2]{\vert}{\vert}%
  {#1\kern0.15ex\delimsize\rangle\delimsize\langle\kern0.15ex\mathopen{}#2}

\DeclarePairedDelimiterX\sandwich[3]{\langle}{\rangle}%
  {#1\,\delimsize\vert\kern0.15ex\mathopen{}#2\kern0.15ex\delimsize\vert\kern0.15ex\mathopen{}#3}

\DeclarePairedDelimiter{\oket}{\vert}{)}

\DeclarePairedDelimiterX\obraket[2]{(}{)}%
  {#1\kern0.15ex\delimsize\vert\kern0.15ex\mathopen{}#2}

\DeclarePairedDelimiterX\oketbra[2]{\vert}{\vert}%
  {#1\kern0.15ex\delimsize)\delimsize(\kern0.15ex\mathopen{}#2}

\DeclarePairedDelimiterX\osandwich[3]{(}{)}%
  {#1\,\delimsize\vert\kern0.15ex\mathopen{}#2\kern0.15ex\delimsize\vert\kern0.15ex\mathopen{}#3}

\newcommand{\myleft}{\mathopen{}\mathclose\bgroup\left}
\newcommand{\myright}{\aftergroup\egroup\right}

\newcommand{\Psym}[1][2]{P_{\mathrm{sym}^{#1}}}

\newcommand{\hhu}{Institute for Theoretical Physics,
	Heinrich Heine University D{\"u}sseldorf, 
	Germany
}
\newcommand{\tuhh}{%
    Institute for Quantum-Inspired and Quantum Optimization,
    Hamburg University of Technology, Germany
}

\title{Closed-form analytic expressions for shadow estimation with brickwork circuits}

\author{Mirko Arienzo}
\affiliation{\hhu}
\affiliation{\tuhh}
\email{mirko.arienzo@tuhh.de}
\author{Markus Heinrich}
\email{markus.heinrich@hhu.de}
\affiliation{\hhu}
\author{Ingo Roth}
\affiliation{Quantum research centre, Technology Innovation Institute, Abu Dhabi, United Arab Emirates}
\email{ingo.roth@tii.ae}
\author[MK]{Martin Kliesch}
\email{martin.kliesch@tuhh.de}
\affiliation{\hhu}
\affiliation{\tuhh}

\date{}    % Activate to display a given date or no date

\begin{document}
\maketitle
%%% ========== PDF meta data =====================
\hypersetup{
	pdfauthor = {Mirko Arienzo, Markus Heinrich, Ingo Roth, and Martin Kliesch}, 
	pdfsubject = {Quantum science},
	pdfkeywords = {%
		quantum, information, computing,
		shadow, tomography, classical,
		Pauli, observable, estimation,
		random, brickwork, layered, circuits, brick,
		tensor network, 
		Clifford, group, 
		unitary, 2, 3, design,
		frame, operator, dual
	}
}
%%% ============================================
\begin{abstract}
	Properties of quantum systems can be estimated using classical shadows, which implement measurements based on random ensembles of unitaries.
	Originally derived for global Clifford unitaries and products of single-qubit Clifford gates,
	practical implementations are limited to the latter scheme for moderate numbers of qubits.
	Beyond local gates, 
	the accurate implementation of very short random circuits with two-local gates is still experimentally feasible and, therefore, interesting for implementing measurements in near-term applications.
	In this work, we derive closed-form analytical expressions for shadow estimation using brickwork circuits with two layers of parallel two-local Haar-random (or Clifford) unitaries. 
	Besides the construction of the classical shadow, our results give rise to sample-complexity guarantees for estimating Pauli observables.
	We then compare the performance of shadow estimation with brickwork circuits to the established approach using local Clifford unitaries and find improved sample complexity in the estimation of observables supported on sufficiently many qubits.
\end{abstract}

%%% ===================================
\section{Introduction}
%%% ===================================
Retrieving information about the state of a quantum system is a long-standing problem in quantum information processing and of central practical importance in quantum technologies.
Full quantum state tomography can recover a complete, precise classical description of the state but requires a large number of state copies \cite{Kliesch2020TheoryOfQuantum, darianoQuantumTomography2003, DonWri15, HaaHarJi15, FlaGroLiu12}, making the protocol feasible only for a very moderate number of qubits.
Nevertheless, for many concrete tasks, complete knowledge of the quantum state is often unnecessary \cite{Aaronson2018ShadowTomography}, and estimation schemes for specific properties are often scalable.

A particularly attractive estimation primitive is nowadays referred to as \emph{shadow estimation} \cite{huangPredictingManyProperties2020,painiApproximateDescriptionQuantum2019}. 
Here, an approximation of a repeatedly prepared unknown quantum state, the so-called \emph{classical shadow}, is constructed from measurements in randomly selected bases. 
In the limit of many bases, this approach allows, in principle, for full state tomography. 
For this reason, classical shadows can be further post-processed to construct estimators for the expectation value of arbitrary sets of observables.
Importantly, for certain random measurement ensembles, rigorous analytical guarantees ensure that precise estimates of expectation values can be evaluated long before one has collected enough measurement statistics for full quantum state tomography.

The original examples with strict guarantees on the sample complexity are, in a sense, two ``extreme'' scenarios:
The first one is characterized by evolving the state with a global random Clifford unitary before performing a basis measurement.
It is particularly suited for predicting global properties; 
for instance, fidelity estimation requires a constant number of samples with this setting.
The second scheme is built on local Clifford unitaries and effectively amounts to perform measurements in random local Pauli bases.
In this case, local properties can often be efficiently estimated \cite{elbenMixedstateEntanglementLocal2020,zhang2021ExperimentalQuantumState,struchalinExperimentalEstimationQuantum2021}.
Moreover, biasing the distribution of local Clifford unitaries to the estimation task at hand can yield further improvements in sample complexity \cite{hadfieldMeasurementsQuantumHamiltonians2022}.

An accurate estimation requires a precise experimental implementation of the random unitaries.
Although more robust variants of shadow estimation exist \cite{chenRobustShadowEstimation2021,kohClassicalShadowsNoise2020}, 
the implementation of global multi-qubit Clifford unitaries on near-term hardware will typically introduce too much noise to be useful for estimation.

Experimentally feasible alternatives, naturally interpolating between the two extreme cases and potentially lowering the sample complexity over local Clifford unitaries, are \emph{short Clifford circuits} \cite{hu2021ClassicalShadowTomography}.
However, finding expressions for classical shadows for random low-depth Clifford circuits is a challenging task.
For instance, the construction by \citet{hu2021ClassicalShadowTomography} involve numerically solving a large system of equations.

In this work, we derive closed-form analytic expressions for the arguably simplest non-trivial circuit construction of classical shadows:
One round of a brickwork circuit consisting of two layers of products of random unitaries.
Besides providing a more direct construction of the classical shadow, these analytic expressions allow us to compare the sample complexity of the circuit construction to the one with local Clifford unitaries.
In particular, we first observe that for Pauli observables, one shall look at pairs of adjacent qubits in the support of such observables and their relative position in the circuit.
Then, we find that the (very short) brickwork shadows outperform the local Clifford ones for Pauli observables supported on sufficiently many qubits of a brickwork circuit.
Conversely, we also observe that local Clifford unitaries yield a lower sample complexity in the case of Pauli observables supported on sparsely distributed qubits in the sense of the brickwork circuit.

The remainder is structured as follows:
Following the observation that the associated measurement channel can be interpreted as a frame (super-)operator \cite{Helsen21EstimatingGate-set} in \cref{sec:classical_shadows},
we work out its matrix representation in the Pauli basis in \cref{sec:brickwork_results}.
In particular, using well-known expressions for the second-moment operator of sufficiently uniform probability measures over the unitary group, we derive recurrence relations for subcircuits that can be analytically solved.
In \cref{sec:discussion}, we identify the regime where the resulting sample complexity outperforms the shadow estimation protocol with the local Cliffords ensemble, and in \cref{sec:numerics} we compare numerically the performance of brickwork and local Cliffords shadows.

\paragraph{Related works.}
During the completion of this work, two other papers on brickwork circuits were published \cite{Akhtar2023scalableflexible, bertoniShallowShadowsExpectation2022}.
Both describe shadows associated with brickwork circuits of arbitrary depth and numerically study the measurement channels associated with such circuits using tensor network techniques.
In particular, 
\citet{Akhtar2023scalableflexible}
apply the formalism based on entanglement features introduced by \citet{buClassicalShadowsPauliinvariant2022} and discusses average case scenario upper bounds on sample complexity based on the locally scrambled shadow norm \cite{hu2021ClassicalShadowTomography}.
A similar discussion, following a probabilistic interpretation of the eigenvalues of the measurement channels, is done by \citet{bertoniShallowShadowsExpectation2022}.
In particular, they provide rigorous upper bounds to the locally scrambled shadow norm for circuits of depth logarithmic in the number of qubits, and find upper bounds to the shadow norm for a class of observables beyond the Pauli case.
In comparison, we only focus on single-round brick-layer circuits but provide analytic expressions for the estimator of Pauli observables.

%%% ===================================
\section{Preliminaries}
%%% ===================================

\subsection{Notation}
We denote the Hilbert-Schmidt inner product by a braket-like notation, namely
\begin{equation}
	\Tr(A^\ad B) \equiv \obraket{A}{B} \quad A, B \in \CC^{\gdim \times \gdim} \, .
\end{equation}
Likewise, the outer product $\oketbra{A}{B}$ denotes the superoperator $C \mapsto \obraket{B}{C} A$.
We parametrize single-qubit Pauli operators by binary vectors $v=(z,x)\in\FF_2^2$ as 
\begin{align}
 W(0,0) &\coloneqq \1, &
 W(0,1) &\coloneqq X, &
 W(1,0) &\coloneqq Z, &
 W(1,1) &\coloneqq Y, 
\end{align}
where $X,Y,Z \in \CC^{2\times 2}$ are the usual Pauli matrices. 
Then, we define the $n$-qubit Pauli operators as tensor products of the single-qubit Pauli operators, indexed by vectors $v = v_1\oplus\dots\oplus v_n \in\FF_2^{2n}$:
\begin{equation}
\label{eq:pauli_operator:def}
	W(v) \coloneqq W(v_1) \otimes \dots \otimes W(v_n) \, .
\end{equation}
For a given vector $v = v_1\oplus\dots\oplus v_n \in\FF_2^{2n}$, we define its \emph{weight vector} as the binary vector $\wt(v)\in\FF_2^n$ such that $\wt(v)_i = 0$ if $v_i = (0,0)$ and $\wt(v)_i = 1$ else.
In other words, $\wt(v)$ has a zero in the $i$th position if and only if $W(v)$ is the identity on the $i$th qubit.
We use the shorthand notation 
\begin{equation}
	\oket{v} \equiv \frac{1}{\sqrt d} W(v)
\end{equation}
for the normalized Pauli operators.
Hence, the set $\{ \oket{v} \}$ denotes the orthonormal Pauli basis in $\CC^{\gdim \times \gdim}$, 
where $\gdim = 2^n$ denotes the dimension of the Hilbert space of $n$ qubits from now on. 

Finally, for any $k\in\mathbb N$, we set $[k]\coloneqq \{1,\dots,k\}$.

\subsection{Classical shadows formalism}
\label{sec:classical_shadows}

In this section, we review the shadow estimation protocol \cite{huangPredictingManyProperties2020} in the language of frame theory (see Ref.~\cite{waldronIntroductionFiniteTight2018} for an introduction to frame theory).
The procedure works as follows: draw unitaries $U\sim\nu$ according to some probability measure $\nu$ on the unitary group $\U(d)$, apply $U$ to the (unknown) state $\rho$, and finally measure in the computational basis $\{E_i \coloneqq \ketbra{i}{i}\}_{i\in[\gdim]}$.
Having obtained outcome $i$, store the \emph{classical snapshot} $(i,U)$. Repeating this primitive yields multiple snapshots $\{i_k, U_k\}_{k=1}^m$.
Finally, given an observable $O$, one evaluates a scalar function $f_O(i,U)$ for each snapshot and takes the empirical average $\hat o = \sum_{k=1}^m f_O(i_k, U_k)$.

Constructing $f_O(i,U)$ as follows ensures that $\hat o$ is an unbiased estimator for the expectation value $\Tr(O\rho)$:
First, one shall require that $\{ E_{i,U} \coloneqq U^\ad E_i U \}$ is a tomographically complete, \ac{POVM} \cite{acharya2021ShadowTomographybased}, i.e.\ for all states $\rho \neq \sigma$ there exists a pair $(i,U)$ such that $\sandwich{i}{U \rho U^\ad}{i} \neq \sandwich{i}{U \sigma U^\ad}{i}$.
This ensures that $\{ E_{i,U} \}$ is a frame \cite{Moran2013, Kliesch2020TheoryOfQuantum}, and the associated measurement channel
\begin{equation}
\label{eq:frame_op}
	\frameop (\rho) \coloneqq
	\sum_{i \in [\gdim]} \EE_{U \sim \nu}
	\oket{E_{i,U}} \obraket{E_{i,U}}{\rho}
	=
	\sum_{i \in [\gdim]} \EE_{U \sim \nu} \sandwich{i}{U \rho U^\ad}{i}\, U^\ad \ketbra{i}{i} U \, 
\end{equation}
has the interpretation as a frame operator.
In particular, $\frameop$ is positive definite, and thus invertible.
Then, $\{\tilde E_{i,U} \coloneqq \frameop^{-1} (E_{i,U})\}$ is the so-called \emph{canonical dual frame}, and we have the following relation
\begin{equation}
	\label{eq:reconstruction_shadow}
	\Tr ( O \rho ) 
	= \Tr ( O \frameop^{-1} \frameop (\rho) ) 
	= \sum_{i} \EE_{U \sim \nu} \obraket{O}{\tilde E_{i,U}} \obraket{E_{i,U}}{\rho}\,. 
\end{equation}
Therefore, the last expression can be interpreted as the expected value of $f_O(i,U) \coloneqq \obraket{O}{\tilde E_{i,U}}$ when sampling $U\sim \nu$ and $i\sim \obraket{E_{i,U}}{\rho}$ and is, thus, the limit of the empirical average over many experimental snapshots.

However, the computation of the canonical dual frame is in general a highly non-trivial task.
Analytical inversion of $S$ is often only possible in special cases where the probability measure $\nu$ is very structured.
For instance, if $\nu$ is the Haar measure on $\U(d)$, or a unitary 2-design, then the POVM $\{E_{i,U}\}$ is a complex projective (state) $2$-design and, thus, forms a tight frame on the subspace of traceless Hermitian matrices.
As a consequence, $\frameop$ is a depolarizing channel and can be readily inverted.
A similar argument can be applied when the unitaries $U$ are drawn Haar-randomly from a subgroup $G\subset\U(d)$ \cite{heinrich2023general}.
More generally, one has to rely on numerical methods which are not only expensive, but may also be numerically unstable since there are no general guarantees on the condition number of $\frameop$.
In principle, the condition number can even be exponentially large \cite{heinrich2023general}.

Under certain conditions, the inversion of $\frameop$ is however drastically simplified:
For instance, if the measure $\nu$ is right-invariant under multiplication with Pauli operators, then $\frameop$ is diagonal in the Pauli basis \cite{buClassicalShadowsPauliinvariant2022}.
This follows from the observation that, in this case, we have $\mathcal{W}(v)^\dagger \frameop \mathcal{W}(v) = \frameop$, where $\mathcal{W}(v) \coloneqq W(v) (\argdot) W(v)^\dagger$, and hence $\frameop$ is invariant under the channel twirl over the Pauli group.
Thus, it is a Pauli channel and, in particular, diagonal in the Pauli basis, which means $\frameop^{-1}$ can be computed via entrywise inversion of the diagonal elements $\osandwich{v}{\frameop}{v}$.
Notice that, for Pauli invariant ensemble without group structure, it is convenient to construct the estimator according to \cref{eq:reconstruction_shadow} instead of the classical shadows $\frameop^{-1}(\rho)$ as in \cite{huangPredictingManyProperties2020}:
For instance, for sparse observables in the Pauli basis \cite{bertoniShallowShadowsExpectation2022}, the estimator can be computed more easily than the classical shadows. 
Indeed, in the latter case, one would rely on the decomposition of $\rho$ in the Pauli basis, which usually involves exponentially many terms.

Finally, if $O=W(v)$ is a Pauli observable (we call this task \emph{Pauli estimation}), the sample complexity of shadow tomography can be bounded for simple circuits.
In particular, if $\frameop$ is diagonal in the Pauli basis, we simply have
\begin{equation}
	\label{eq:pauli_estimation}
 	f_{W(v)}(i,U) = \frac{1}{\osandwich v \frameop v} \obraket {W(v)} {E_{i,U}}\,.
\end{equation}
Note that this expression features only a single diagonal element of the frame operator independent of $i$ and $U$.
The sample complexity of the corresponding mean estimator $\hat w(v)$ can be controlled using the variance of $f_{W(v)}(i,U)$ which can be shown to be dominated by $\osandwich v \frameop v^{-1}$ \cite{buClassicalShadowsPauliinvariant2022}.
Chebyshev's inequality then ensures that the mean estimator is $\epsilon$-precise using $O(\osandwich v \frameop v^{-1} \epsilon^{-2}\delta^{-1})$ many snapshots with probability $1-\delta$.
Note that a Hoeffding bound here yields a worse bound scaling as $\osandwich v \frameop v^{-2}$.
If the expectation values of `many' observables are to be estimated at once, it may be beneficial to use the \emph{median-of-mean estimator} with sample complexity depending only logarithmically on $\delta$ \cite{huangPredictingManyProperties2020}.

In general, however, it is not easy to find strict guarantees for the sample complexity, since it is hard to analytically bound the variance, even for different classes of Pauli invariant measures.
In these cases, one can rely on the weaker notion of \emph{locally scrambled shadow norm} \cite{hu2021ClassicalShadowTomography, Akhtar2023scalableflexible, bertoniShallowShadowsExpectation2022}, which can be interpreted as the average variance over all states.
In particular, since the variance is linear in the state $\rho$, the locally scrambled shadow norm thus quantifies the performance when $\rho$ is the completely mixed state.

%%%% =====================
\section{The brickwork circuit: analytical results}
\label{sec:brickwork_results}
%%%% =====================
We assume for simplicity that the number of qubits is even and consider one round of a one-dimensional \ac{BW} circuit built in the following way: a first layer of $n/2$ two-local Haar random unitaries is applied to qubits $(2i-1, 2i)$ for $i \in [n/2]$. 
The second layer, built in the same way but shifted by one position, applies Haar random unitaries to qubits $(2i, 2i+1)$.
Here, we consider two cases, see also Figure~\ref{fig:brickwork_periodic}.
First, the second layer has periodic boundary conditions such that qubits $n+1$ and $1$ are identified, and consequently, the $n/2$th random unitary acts on the qubit pair $(n,1)$.
Second, we treat the case of open boundary conditions, where the second layer does not act on the first and the $n$th qubit.
In practice, it can be more convenient to draw unitaries from a unitary $2$-design, such as the Clifford group (which, for qubits systems, is even a $3$-design \cite{Zhu15, Web15}).
Indeed, implementing Haar-random unitaries is very hard \cite{knill_1995} and, moreover,
employing Clifford unitaries ensures one can classically post-process shadows efficiently \cite{gottesmanHeisenbergRepresentationQuantum1998}.

\begin{figure}[tb]
	\centering
	\includegraphics[width=\textwidth]{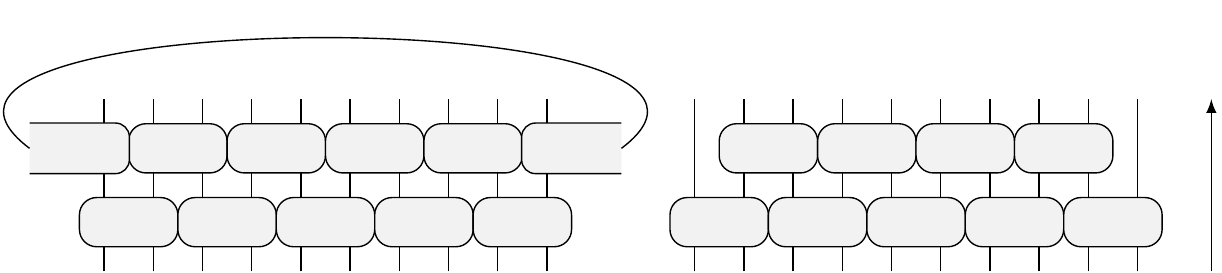}
	\caption{
		Brickwork circuits acting on $n=10$ qubits. 
		The left and right figures show periodic and open boundary conditions, respectively, and the arrow indicates the direction in which the circuit acts on quantum states. 
		For both of them, the first layer is composed of $n/2$ two-qubit Haar random unitaries acting on qubits $(2i-1, 2i), \ i \in [5]$, and the second layer is shifted by one position.
		On the left, the bricks in the second layer acts on qubits $(2i, 2i+1) \, , i \in [5]$, with the periodic identification $n+1 = 1$.
		On the right, the second layer acts on qubits $(2i, 2i+1) \, , i \in [4]$, leaving the first and the $n$th qubit untouched.
	}
	\label{fig:brickwork_periodic}
\end{figure}

In the following, we derive analytical results for the frame operator of random brickwork circuits with open and periodic boundary conditions.
Both \ac{BW} circuit ensembles are clearly (left and right) invariant under tensor products of single-qubit unitaries, in particular they are right-invariant under Pauli operators.
By the preceding discussion in \cref{sec:classical_shadows}, we thus know that the frame operator $\frameop$ is diagonal in the Pauli basis.
It is thus sufficient to compute the matrix elements $\osandwich{v}{\frameop}{v}$ for all $v\in\FF_2^{2n}$.
Moreover, both \ac{BW} circuit ensembles are also invariant under local Clifford unitaries, i.e.~tensor products of single-qubit Clifford gates.
This implies that $\osandwich{v}{\frameop}{v}$ is invariant under the exchange of $X$, $Y$, and $Z$ operators, and hence depends only on the weight vector $\wt(v)$.
As we show shortly, $\osandwich{v}{\frameop}{v}$ is in fact determined
by non-vanishing pairs of elements in $\wt(v)$ corresponding to a brick in the second layer, and by their positions in the circuit.
To make this precise, we have to introduce some definitions.

\begin{figure}[tb]
\centering
\includegraphics[width=0.9\textwidth]{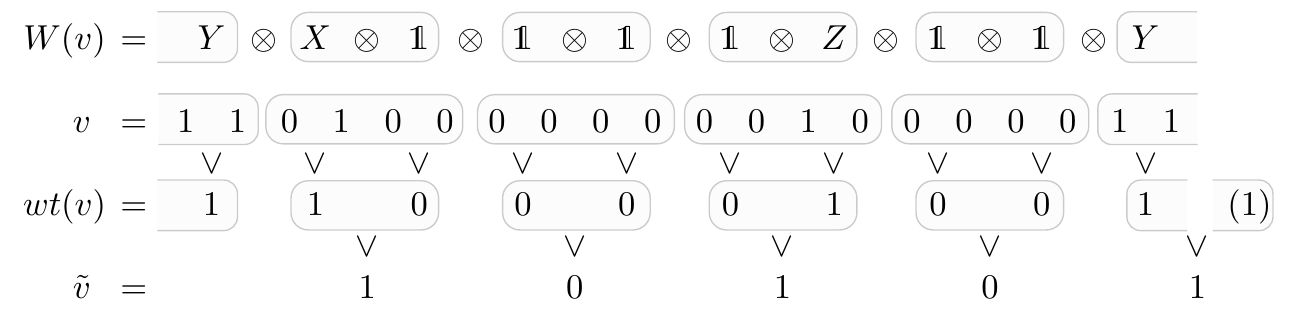}
 \caption{ 
	 Example how the vector of supported bricks is computed from a 10-qubit Pauli operator.
	 The structure of the shaded bricks is the one of the second layer of the circuit.
	 First, the Pauli operator is transformed into its binary representation $v\in\FF_2^{20}$.
	 We apply a logical or ($\vee$) per qubit to compute the weight vector $\wt(v)$.
	 Subsequently, this procedure is repeated for qubit pairs $(2i,2i+1)$ and yields the vector of supported bricks $\tilde{v}$.
	 For periodic boundary conditions, the last entry of $\tilde v$ is computed between the last and first entry of $\wt(v)$ (here depicted by appending the first entry at the end in parentheses).
	 The brickwork support of this example is $\suppBW(v) = \{1, 3, 5\}$, while its partition into local factors is $\partBW(v) = (1, 2)$.
 }
 \label{fig:example-tilde-v}
\end{figure}

Let us consider a Pauli string $v = v_1\oplus\dots\oplus v_n \in \FF_2^{2n}$.
Roughly speaking, a brick is identified by a pair of two adjacent qubits, and it is in the support of $v$ if at least one of the qubits is in the support of $v$.
More formally, we define the vector of \emph{supported bricks} as
\begin{equation}
	\tilde v = (\tilde v_1, \dots, \tilde v_{n/2}) \in \FF_2^{n/2} \, , \quad
	\tilde v_i \coloneqq \wt(v)_{2i} \vee \wt(v)_{2i+1} \, , \quad
	i \in [n/2] \, ,
\end{equation}
where $x \vee y$ is the \emph{logical or} between two bits $x,y\in\FF_2$, i.e.~$x \vee y = 1$ if $x=1$ or $y=1$, and $0$ else.
The last entry of $\tilde v$ is defined according to the boundary conditions of the second layer, in particular $\tilde v_{n/2} = \wt(v)_n \vee \wt(v)_1$ for periodic boundary conditions, and $\tilde v_{n/2} \equiv 0$ for open boundary conditions, see also \cref{fig:example-tilde-v} for an explicit example how $\tilde v$ is computed.
We say that the $i$th brick in the second layer, with $i \in [n/2]$, \emph{is in the support of} $v$ if $\tilde v_i = 1$. 
Then, one can define the \emph{brickwork support} of $v \in \FF_2^{2n}$ as $\suppBW(v) \coloneqq \{ i \mid \tilde v_i \neq 0\}$.
In the following, however, it will be equally important to keep track of sequences of consecutive supported bricks in the circuit.
Hence, we introduce the following notation:
A \emph{one-component} of $\tilde v$ is a maximal tuple of consecutive ones in $\tilde v$, where ``consecutive'' is again meant w.r.t.~the boundary conditions of the \ac{BW} circuit.
Then, we define the \emph{partition of the brickwork support} $\partBW(v)$ to be the integer sequence given by the (non-unique) sizes of the one-components of $\tilde v$.
For instance, if we have periodic boundary conditions and $\tilde v = (1, 0, 1, 0, 1)$ as in \cref{fig:example-tilde-v}, then $\partBW(v) = (1, 2)$.
Note that the maximal number of consecutive ones is $n/2-1$ and $n/2$ for open and periodic boundary conditions, respectively. 

We can now state our main result: 

\begin{theorem}
	\label{thm:main}
	Let $\frameop$ be the frame operator associated with one round of a two-local brickwork circuit with open or periodic boundary conditions in the second layer.
	Then, $\frameop$ is diagonal in the Pauli basis, and for $v \in \FF^{2n}_2$
	\begin{equation}
		\label{eq:vSv_main}
		\osandwich{v}{\frameop}{v}
		=
		\begin{cases}
			\frameopPeriodic(n) \, , & \quad \partBW(v) = ( n/2 ) \, , \\
			\prod_{l \in \partBW(v)} \frameopOpen(2l+2) \, , & \quad \text{otherwise} \, ,
		\end{cases}
	\end{equation}
	where, for any $m \in \NN$ even,
	\begin{align}
		\label{eq:vSv_final_periodic}
		\frameopPeriodic(m)
		& =
		\frac{\left( \sqrt{41} + 5 \right)^{m/2} + (-1)^{m/2} \left( \sqrt{41} - 5 \right)^{m/2}}{(5 \sqrt 2)^m} \, , \\
		\label{eq:vSv_linear_final}
		\frameopOpen(m)
		& =
		\frac{5}{2 \sqrt{41}} \frac{\left( 25-3\sqrt{41} \right) \left( \sqrt{41} + 5 \right)^{m/2} + (-1)^{m/2+1} \left( 25+3\sqrt{41} \right) \left( \sqrt{41} - 5 \right)^{m/2}}{\left(5 \sqrt{2}\right)^m} \, .
	\end{align}
\end{theorem}
We provide a proof for the theorem in \cref{sec:proof_theorem}.

Let us briefly comment on the interpretation of the matrix elements of $\frameop$.
These values, determined by the elements of $\partBW(v)$, are associated with different topologies of the effective \ac{BW} circuit.

First, notice that the case $\partBW(v) = (n/2)$ can occur for periodic boundary conditions only and corresponds to all bricks being in the support of $v$.
In particular, for open boundary conditions, the second case in \cref{eq:vSv_main} always applies.

Next, let us motivate the second case in \cref{eq:vSv_main}.
Concretely, let us first assume $\partBW(v) = (n/2-1)$.
In the case of open boundary conditions, this assumption corresponds to all bricks being in the support of $v$, since $\tilde v_{n/2} = 0$ by definition.
Likewise, this situation occurs in the \ac{BW} circuit with periodic boundary conditions whenever there exists exactly one $i \in [n/2]$ such that $\tilde v_i = 0$, see Figure~\ref{fig:brickwork_open_examples}.
Then, we can make two observations:
First, the topology of the effective circuit changes from periodic to open boundary conditions.
Second, the effective circuit is equivalent --up to reordering of qubits on which it acts-- to the fully supported circuit with open boundary conditions described before, which is depicted in Figure \ref{fig:brickwork_periodic}.

Suppose now we have a \ac{BW} circuit with open boundary conditions, and there exists another index $i$ such that $\tilde v_i = 0$.
Then, two cases can occur: 
Either, (a), $i=1$ or $i=n/2-1$, which implies $\partBW(v) = (n/2 - 2)$, and we simply obtain a \ac{BW} circuit with open boundary conditions on $n-4$ qubits. 
Otherwise, (b), $\partBW(v) = ( i-1, n/2 - i -1 )$ and the \ac{BW} circuit again factorizes into two independent \ac{BW} circuits with open boundary conditions, acting on $2i$ and $n - 2i$ qubits respectively.

In general, the effective circuit splits into as many independent \ac{BW} circuits with open boundary conditions as the number of elements in $\partBW(v)$, and the diagonal elements of $\frameop$ are given by products of different contributions as in \cref{eq:vSv_main}.
These elements also determine the number of qubits on which these subcircuits act, see Figure~\ref{fig:brickwork_open_examples} for an example with $\abs{\partBW(v)} = 2$.

\begin{figure}
	\centering
	\includegraphics[width=\textwidth]{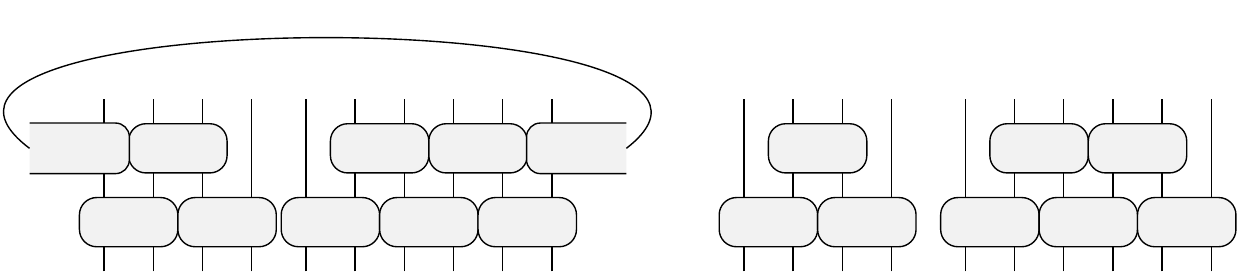}
	\caption{
		Effective brickwork circuits associated with non-fully supported Pauli operators in the case $n=10$.
		On the left, the effective \ac{BW} circuit associated with a $v \in \FF_2^{20}$ such that $\tilde v_2 = 0$ and $\partBW(v) = (4)$.
		In this case, open boundary conditions apply and the circuit is topologically equivalent to the right one in Figure \ref{fig:brickwork_periodic}.
		On the right, the effective \ac{BW} circuit associated with $v \in \FF_2^{20}$ such that $\tilde v_2, \tilde v_5 = 0$, which implies $\suppBW(v) = \{1, 3, 4\}$ and $\partBW(v) = (1,2)$.
		In this case, the circuit is the product of two smaller subcircuits with open boundary conditions.
		In particular, the subcircuits are defined on $4$ and $6$ qubits, respectively.
	}
	\label{fig:brickwork_open_examples}
\end{figure}

The frame operator provides, as proved in Ref.~\cite{buClassicalShadowsPauliinvariant2022}, a bound on the variance of the Pauli estimation task.
Here, the variance of a single sample is given as
\begin{equation}
	\varBW(v, \rho) 
	\coloneqq 
	\sum_i \EE_{U\sim \mu} 
	f_{W(v)}(i,U)^2
	\obraket{E_{i,U}}{\rho} - \obraket{W(v)}{\rho}^2 \, ,
\end{equation}
where $f_{W(v)}$ is defined as in \cref{eq:pauli_estimation}.
Then, the following result holds:
\begin{proposition}[{\cite[Prop.~3]{buClassicalShadowsPauliinvariant2022}}]
	\label{prop:variance_upper_bound}
	For any state $\rho$, estimate $W(v)$ using \ac{BW} shadows.
	Then, $\varBW$ depends only on $v \in \FF_2^{2n}$, and
	\begin{equation}
		\varBW(v,\rho) \equiv \varBW(v)
		\leq
		\frac{1}{\osandwich v S v} \, .
	\end{equation}
\end{proposition}
We remark that the latter holds in general for any ensemble invariant under Pauli multiplication.

In \cref{app:variance} we provide an alternative proof for circuits with periodic boundary conditions which holds for unitary $3$-designs.

%%%% =====================
\subsection{Proof of \texorpdfstring{\cref{thm:main}}{Theorem~\ref{thm:main}}}
\label{sec:proof_theorem}
%%%% =====================
In this section, we provide a proof of \cref{eq:frame_op}.
It goes through the following steps:
First, in \cref{lem:frame_op_diagonal}, we will prove that such eigenvalues are determined by $\partBW(v)$ for any $v \in \FF_2^{2n}$.
Those eigenvalues are associated with different, effective \ac{BW} circuits.
Exploiting the structure of the \ac{BW} circuit, we can `split' the action of the layers in two separate ($2$-local) group twirls, which can be evaluated using standard results in the computation of moment operators (see \cref{app:auxiliary_results} for a quick introduction on these techniques and the main facts of our interest).
Next, in \cref{lem:recursive_relations_periodic,lem:recursive_relations_open}, using tensor network techniques, we reduce the problem of finding such eigenvalues to two different systems of recurrence relations associated with \ac{BW} circuits with periodic and open boundary conditions, respectively.
The latter admit closed-form analytic solutions, which lead to the explicit expression of $\osandwich{v}{\frameop}{v}$.

In the following, given an operator $A \in \L(\CC^{2} \otimes \CC^{2})$, we denote by $A_{(2)} \equiv A \otimes A \in \L(\CC^{4} \otimes \CC^{4})$ the operator acting on two copies of two qubit ``sites''.
In particular, we will extensively use the operator $\flip_{(2)} \equiv \flip \otimes \flip$, where $\flip$ is the flip operator which swaps tensor factors of $\CC^{2} \otimes \CC^{2}$.
The action of $\flip$ is also depicted in \cref{fig:flip_2copies_example}.
\begin{figure}[tb]
	\centering
	\includegraphics[width=0.6\textwidth]{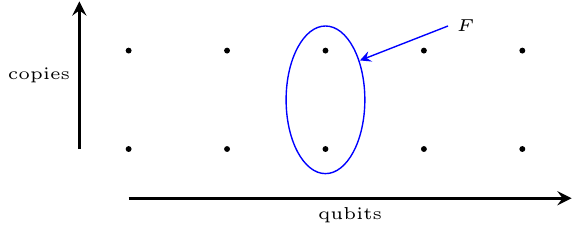}
	\caption{
		Flip operator swapping tensor factors of two copies of the third qubit site.
	}
	\label{fig:flip_2copies_example}
\end{figure}

As we observed in \cref{sec:classical_shadows}, the frame operator is diagonal in the Pauli basis \cite{buClassicalShadowsPauliinvariant2022} for Pauli invariant measures.
However, since unitaries in the \ac{BW} circuit are Haar random, we can characterize the matrix elements of the frame operator exploiting some known results about the second moment operator \cite{Kliesch2020TheoryOfQuantum}, which we summarize in \cref{lem:Haar_second_moment} in \cref{app:auxiliary_results}.
As a side note, this also implies that, in practice, we can draw unitaries from any unitary $2$-design instead \cite{dankertEfficientSimulationRandom2005,GroAudEis07}, such as the Clifford group.

Then, the following holds:
\begin{lemma}
	\label{lem:frame_op_diagonal}
	Let $\frameop$ be the frame operator associated with one round of a two-local brickwork circuit with periodic or open boundary conditions in the second layer.
	Then, $\frameop$ is diagonal in the Pauli basis, and
	\begin{equation}
		\label{eq:vSv_lem}
		\osandwich{v}{\frameop}{v}
		=
		\begin{cases}
			\frac{1}{(10 \sqrt{3})^n} \tper(n) \, , & \quad \partBW(v) = ( n/2 ) \\
			\frac{1}{15^{\abs{\tilde v}}}
			\prod_{l \in \partBW(v)} \frac{\topen(2l+2)}{(2 \sqrt{5})^{2l+2}} \, , & \quad \text{otherwise}
		\end{cases} \, ,
	\end{equation}
	where $\abs{\tilde v}$ is the Hamming weight of $\tilde v$, and
	\begin{align}
		\label{eq:tper(n)_def}
		\tper(n)
		& \coloneqq
		\Tr\myleft[ (\1+\flip_{(2)})^{\otimes n/2} \ D_{(2)} \left( 4 \flip_{(2)} - \1 \right)^{\otimes n/2} D_{(2)}^{-1} \myright] \, , \\
		\label{eq:topen(n)_def}
		\topen(n) 
		& \coloneqq 
		\Tr\myleft[ (\1 +\flip_{(2)})^{\otimes n/2} 
		\left\{ \1_4 \otimes \left( (4 \flip_{(2)} -\1)^{\otimes (n/2-1) } \right) \otimes \1_4 \right\} \myright] \, ,
	\end{align}
	with $D\ket{\psi_1}\otimes\dots\otimes\ket{\psi_n} \coloneqq \ket{\psi_2}\otimes\dots\otimes\ket{\psi_n}\otimes\ket{\psi_1}$ being a cyclic shift operator.
\end{lemma}
\begin{proof}
	Let $\mu$ be a probability measure on the \ac{BW} circuit.
	Hence, $\mu$ is the product of probability measures $\mu_{i_j}$, where $i = 1, \dots, n/2$, and $j = 1,2$.
	In other terms, the operator corresponding to the $i_j$th brick is sampled independently from all the others.
	Then, given $u, v \in \FF_2^{2n}$, we have
	\begin{equation}
		\label{eq:Svv_general}
		\begin{aligned}
			\osandwich u S v 
			& = 
			\frac{1}{\gdim}
			\sum_{i} \EE_{U \sim \mu} \obraket{u}{E_{i,U}} \obraket{E_{i,U}}{v} \\
			& =
			\frac{1}{\gdim}
			\sum_i \EE_{U \sim \mu} \Tr[W(u)^\dagger U \ketbra ii U^\dagger]\,  \Tr[U \ketbra ii U^\dagger W(v)]
			\\
			& =
			\EE_{U \sim \mu} \sandwich{0}{U^{\otimes 2\dagger} W(u) \otimes W(v) U^{\otimes 2}}{0} \, .
		\end{aligned}
	\end{equation}
	Consider now the following factorization of $W(v)$:
	\begin{equation}
		W(v) = W(v_{2,3}) \otimes \dots \otimes W(v_{n,1}) \, ,
	\end{equation}
	where, for each $i \in [n/2]$,
	\begin{equation}
		v_{2i, 2i+1} \equiv v_{2i} \oplus v_{2i+1} \in \FF_2^2 \oplus \FF_2^2 \, ,
	\end{equation}
	and each $W(v_{2i, 2i+1})$ is as in \cref{eq:pauli_operator:def}.
	Moreover, writing $U = D U_2 D^\ad U_1$, where $U_i$ is the tensor product of two-local Haar random unitaries, it follows
	\begin{equation}
		\begin{aligned}
			\osandwich u S v 
			&=
			\,\Tr\myleft[ \EE_{U_1} \myleft[ U_1^{\otimes 2} \ketbra 00 U_1^{\otimes 2\dagger} \myright] \, 
			D_{(2)} \, \EE_{U_2} \myleft[ U_2^{\otimes 2\dagger} D_{(2)}^\ad W(u) \otimes W(v) D_{(2)} U_2^{\otimes 2} \myright] \, D_{(2)}^\ad \myright] 
			\, .
		\end{aligned}
	\end{equation}
	Hence, by \cref{lem:Haar_second_moment} in \cref{app:auxiliary_results}, we have
	\begin{align}
		\EE_{U_1} U_1^{\otimes 2} \ketbra{0}{0}^{\otimes 2} U_1^{\otimes 2 \ad} 
		& =
		\frac{1}{10^{n/2}} \Psym[2]^{\otimes n/2} \, , \\
		\EE_{U_2} U_2^{\otimes 2} D_{(2)}^\ad W(v) \otimes W(u) D_{(2)} U_2^{\otimes 2 \ad} 
		& =
		\delta_{u,v}
		\bigotimes_{i \in [n/2]} Q_{\tilde v_i} \, ,
	\end{align}
	where
	\begin{equation}
		\label{eq:Q_v}
		Q_{\tilde v_i}
		\coloneqq 
		\begin{cases}	
			\1 \quad &\text{if } \tilde v_i = 0 \, , \\
			\frac{1}{15} \left( 4 \flip_{(2)}-\1 \right)  \quad &\text{otherwise} \, .
		\end{cases}
	\end{equation}
	Therefore, writing $\Psym[2] = \frac{1}{2} \left( \1 + \flip_{(2)} \right)$,
	\begin{equation}
		\label{eq:vSv:general_main}
		\begin{aligned}
			\osandwich{v}{\frameop}{v}
			& 
			=
			\frac{1}{(2 \sqrt{5})^{n}}
			\Tr\myleft[ (\1+\flip_{(2)})^{\otimes n/2} \ D_{(2)} \, 
			\bigotimes_{i \in [n/2]} Q_{\tilde v_i} \, D_{(2)}^{-1} \myright]
		\end{aligned} \, .
	\end{equation}
	Finally, we distinguish cases for the latter according to $\partBW(v)$.
	In particular, if $\partBW(v) = ( n/2 )$, then $Q_{\tilde v_i} = \frac{1}{15} \left( 4 \flip_{(2)} - \1 \right) \ \forall i = 1, \dots, n/2$, and \cref{eq:tper(n)_def} reads immediately from \cref{eq:vSv:general_main}.
	Next, assume $\partBW(v) = (n/2-1)$.
	In particular, due to invariance under translations of bricks, we can assume without loss of generality that $\tilde v_{n/2} = 0$, meaning $Q_{\tilde v_{n/2}} = \1$.
	This yields
	\begin{equation}
		\label{eq:vSv_n/2-1}
		\osandwich{v}{\frameop}{v}
		=
		\frac{1}{(2 \sqrt{5})^n} \frac{1}{15^{n/2-1}} 
		\Tr \left[
			\left(\1+\flip_{(2)}\right)^{\otimes n/2} \1 \otimes \left(4\flip_{(2)}-\1\right)^{\otimes n/2-1} \otimes \1
		\right] \, .
	\end{equation}
	Consider now there exists $i \in [n/2-1]$ such that $\tilde v_i = 0$.
	Then, we distinguish between two cases.
	If $i = 1$ or $i = n/2-1$, then $\osandwich{v}{\frameop}{v}$ is still given by an expression that is morally equivalent to \cref{eq:vSv_n/2-1} up to obvious modifications determined by $\partBW(v) = (n/2-2)$.
	More specifically, we have
	\begin{equation}
		\begin{aligned}
			\osandwich{v}{\frameop}{v} 
			& = 
			\frac{1}{(2\sqrt{5})^{n}} \frac{1}{15^{n/2-2}}
			\Tr \left[ \1+\flip_{(2)} \right]
			\Tr \left[
				\left(\1+\flip_{(2)}\right)^{\otimes n/2-1} \1 \otimes \left(4\flip_{(2)}-\1\right)^{\otimes n/2-2} \otimes \1
			\right] \\
			& = 
			\frac{1}{(2\sqrt{5})^{n-2}} \frac{1}{15^{n/2-2}} \topen(n-2) \, .
		\end{aligned}
	\end{equation}
	On the other hand, if $i \in \{2,\dots,n/2-2\}$, we have $\partBW(v) = (i-1, n/2-1-i)$, and the circuit splits into two subcircuits, yielding
	\begin{equation}
		\begin{aligned}
			\osandwich{v}{\frameop}{v}
			& = 
			\frac{1}{(2\sqrt{5})^n}
			\Tr \left[
				\left(\1+\flip_{(2)}\right)^{\otimes i} \1 \otimes \left(4\flip_{(2)}-\1\right)^{\otimes (i-1)} \otimes \1
			\right] \\
			& \times
			\Tr \left[
				\left(\1+\flip_{(2)}\right)^{\otimes n/2-i} \1 \otimes \left(4\flip_{(2)}-\1\right)^{\otimes n/2-1-i} \otimes \1
			\right] \\
			& = 
			\frac{1}{(2\sqrt{5})^n} \topen(2i) \topen(n-2i) \, .
		\end{aligned}
	\end{equation} 
	All other cases follow from analogous considerations.
\end{proof}

Note that the traces in the latter expression have two main contributions. The first one, which is proportional to the projector on the symmetric subspace $\Psym[2]$, comes from scrambling $E_i$ with the first layer of the \ac{BW} circuit, and it is independent of $v$.
The second layer, on the other hand, acts on $W(u) \otimes W(v)$, and the result of the scrambling for each pair of qubits is an operator that depends on $v$.
This means that \emph{effectively} the second layer determines whether the circuit factorizes at a given position, and the number of qubits on which each subcircuit is defined is determined by the corresponding first layer of random unitaries.

The next couple of technical results will give a way to evaluate the traces appearing in the previous lemma. 
The core steps of the proofs are most conveniently presented in terms of tensor network diagrams and deferred to \cref{app:tensor_networks}.
\begin{lemma}
	\label{lem:recursive_relations_periodic}
	Let $t_1(n), t_2(n), t_3(n)$ defined as follows:
	\begin{align}
		\label{eq:t1_periodic}
		t_1(n)
		& \coloneqq
		\Tr \left[ 
			\left(
			\1 \otimes \left( \1 + \flip_{(2)} \right)^{\otimes n/2-1} \otimes \1
			\right)
			\left(
			4 \flip_{(2)} -\1
			\right)^{\otimes n/2} 
		\right] \, , \\
		t_2(n)
		\label{eq:t2_periodic}
		& \coloneqq
		\Tr \left[ 
		\left(
		\flip \otimes \left( \1 + \flip_{(2)} \right)^{\otimes n/2-1} \otimes \flip
		\right)
		\left(
		4 \flip_{(2)} -\1
		\right)^{\otimes n/2} 
		\right] \, , \\
		t_3(n)
		\label{eq:t3_periodic}
		& \coloneqq
		\Tr \left[ 
		\left(
		\1 \otimes \left( \1 + \flip_{(2)} \right)^{\otimes n/2-1} \otimes \flip
		\right)
		\left(
		4 \flip_{(2)} -\1
		\right)^{\otimes n/2} 
		\right] \, .
	\end{align} 
	Then, $\tper(n) = t_1(n) + t_2(n)$ and the following system of recursive relations hold true:
	\begin{equation}
		\label{eq:t1_t2_t3_periodic_recursive}
		\begin{cases}
			t_1(n) = 24 \, t_3(n-2) \\
			t_2(n) = 24 \, t_3(n-2) + 60 \, t_2(n-2) \\
			t_3(n) = 24 \, t_1(n-2) + 60 \, t_3(n-2)
		\end{cases} \, ,
		\quad n \geq 2 \, , \quad n=0 \mod 2 \, ,
	\end{equation}
	with the following base conditions:
	\begin{equation}
		\label{eq:t1_t2_t3_periodic_base}
		\begin{cases}
			t_1(2) = 0 \\
			t_2(2) = 60 \\
			t_3(2) = 24
		\end{cases} \, .
	\end{equation}
\end{lemma}
\begin{proof}
	The fact that $\tper(n) = t_1(n) + t_2(n)$ is clear from the definition of $t_1$ and $t_2$.
	Relations \eqref{eq:t1_t2_t3_periodic_recursive} and \eqref{eq:t1_t2_t3_periodic_base} are proved in \cref{app:tensor_networks}.
\end{proof}
\begin{lemma}
	\label{lem:recursive_relations_open}
	Let $t_1(n), t_2(n)$ be defined as follows:
	\begin{align}
		\label{eq:t1_open}
		t_1(n) 
		& \coloneqq
		\Tr\myleft[ \left\{ \1_4 \otimes (\1 + \flip_{(2)})^{\otimes (n/2-1)}\right\}
		\left\{ (4 \flip_{(2)} -\1)^{\otimes (n/2-1) } \otimes \1_4 \right\}  \myright] \, , \\
		\label{eq:t2_open}
		t_2(n) 
		& \coloneqq
		\Tr\myleft[ \left\{  \flip  \otimes (\1+\flip_{(2)})^{\otimes (n/2-1)}\right\}
		\left\{ (4 \flip_{(2)} -\1)^{\otimes (n/2-1) } \otimes \1_4 \right\}  \myright] \, .
	\end{align}
	Then, $\topen(n) = 4 t_1(n) + 2 t_2(n)$, and the following recursive relations hold true:
	\begin{equation}
		\label{eq:t1_t2_open_recursive}
		\begin{cases}
			t_1(n) = 24 \, t_2(n-2) \\
			t_2(n) = 24 \, t_1(n-2) + 60 \, t_2(n-2) 
		\end{cases}
		\, \quad n \geq 4 \, , \quad n=0 \mod 2 \, ,
	\end{equation}
	with the following base conditions:
	\begin{equation}
		\label{eq:t1_t2_open_base}
		\begin{cases}
			t_1(4) = 48 \\
			t_2(4) = 216	
		\end{cases} \, .
	\end{equation}
\end{lemma}
\begin{proof}
	First, observe that
	\begin{equation}
		\begin{aligned}
			\topen(n) 
			&\coloneqq 
			\Tr\myleft[ (\1+\flip_{(2)})^{\otimes n/2} 
			\left\{ \1 \otimes \left(4 \flip_{(2)} -\1\right)^{\otimes (n/2-1) } \otimes \1 \right\} \myright]
			\\
			&=
			\Tr\myleft[ \myleft\{ \Tr_{1} \myleft[ (\1+\flip_{(2)}) \myright] \otimes \myleft[ (\1+\flip_{(2)})^{\otimes (n/2-1)} \myright] \myright\}
			\left\{ (4 \flip_{(2)} -\1)^{\otimes (n/2-1) } \otimes \1 \right\}  \myright]
			\\
			&=
			\Tr\myleft[ \left\{ (4 \1 + 2 \flip ) \otimes (\1+\flip_{(2)})^{\otimes (n/2-1)}\right\}
			\left\{ (4 \flip_{(2)} -\1)^{\otimes (n/2-1) } \otimes \1 \right\}  \myright] \\
			&=
			4 \, t_1(n) + 2 \, t_2(n) \, .
		\end{aligned}
	\end{equation}
	Relations \eqref{eq:t1_t2_open_recursive} and \eqref{eq:t1_t2_open_base} are proved in \cref{app:tensor_networks}.
\end{proof}
\begin{proof}[Proof of \cref{thm:main}]
	As discussed before, $\frameop$ is diagonal in the Pauli basis, and we only need to characterize its diagonal elements $\osandwich{v}{\frameop}{v}$, which, by \cref{lem:frame_op_diagonal}, are determined by %the structure of 
	$\partBW(v)$.
	In the first case, when $\partBW(v) = (n/2)$, the circuit retains periodic boundary conditions, and $\osandwich{v}{\frameop}{v}$ is proportional to $\tper(n)$ according to \cref{eq:vSv_lem}.
	By \cref{lem:recursive_relations_periodic}, $\tper(n)$ can be expressed as the sum of two terms that can be calculated recursively using the system of recurrence relations \eqref{eq:t1_t2_t3_periodic_recursive}, and one can check that the solution of this system is given by
	\begin{equation}
		\label{eq:recurrence_rel_periodic_solution}
		\tper(n) = 
		6^{n/2} \left[ \left(\sqrt{41}+5\right)^{n/2} + (-1)^{n/2} \left(\sqrt{41} - 5\right)^{n/2} \right] .
	\end{equation}
	This solution can be found with a computer algebra system, or, by hand, using the $Z$-transform \cite{oppenheim1999discrete}, which maps recurrence relations to algebraic equations.
	Inserting \cref{eq:recurrence_rel_periodic_solution} into \cref{eq:vSv_lem} then shows \cref{eq:vSv_final_periodic}.
	
	Otherwise, $\partBW(v)$ determines the factorization into (possibly many) subcircuits with open boundary conditions.
	In particular, each entry $l \in \partBW(v)$ determines a (factorized) subcircuit acting on $2l+2$ qubits.
	Each such subcircuit evaluates up to a multiplicative constant to $\topen(2l+2)$, that, by \cref{lem:recursive_relations_open} fulfills the recurrence relations \eqref{eq:t1_t2_open_recursive}.
	These also admit a closed-form solution that can be found with the same techniques and is given as
	\begin{equation}
		\label{eq:recurrence_rel_open_solution}
		\topen(m)
		=
		\frac{6^{m/2}}{6 \sqrt{41}} 
		\left[
		\left(25-3\sqrt{41}\right)\left(\sqrt{41}+5\right)^{m/2} + (-1)^{m/2+1}\left(25+3\sqrt{41}\right)\left(\sqrt{41}-5\right)^{m/2}
		\right] ,
	\end{equation}
	for any $m \equiv 2l+2$.
	This shows \cref{eq:vSv_linear_final} for each subcircuit.
\end{proof}

As a final remark, observe that the proof of the theorem can be generalized to systems of arbitrary prime or power of prime local dimension.
In particular, redefining $\tilde v$ according to the local dimension, \cref{eq:vSv:general_main} holds true with obvious modifications for any prime $\ldim$, and the same holds for the traces \cref{eq:tper(n)_def,eq:topen(n)_def}.

Finally, one may wonder whether it is possible to find analytical expressions for the frame operator associated with circuits with more layers.
However, in this case, splitting the scrambling over multiple layers is more involved, since non-trivial `intertwinings' between layers occur.
This implies that the analytical contraction of the corresponding tensor network is more difficult compared to the calculations of \cref{app:tensor_networks}, and one might only resort to numerical methods to evaluate the frame operator \cite{Akhtar2023scalableflexible, bertoniShallowShadowsExpectation2022}.

%%%% ====================
\section{Discussion and comparison with local Clifford circuits}
\label{sec:discussion}
%%%% ====================

Given the closed analytic expressions for the frame operator associated with the \ac{BW} circuit, we can now compare the performance with the \acp{LC} ensemble.

For \acp{LC}, the variance is exponential in the weight of the Pauli observable \cite{huangPredictingManyProperties2020}.
More precisely, for any $v \in \FF_2^{2n}$ we define the \emph{local Cliffords support} as the set of weighted sites of $v$, namely $\suppLC(v) \coloneqq \{i \mid \wt(v)_i \neq 0\}$.
Then, since the \acp{LC} ensemble is clearly invariant under multiplication with Pauli operators, one can apply \cref{prop:variance_upper_bound} to get a bound on the variance.
In particular, $\varLC(v) \leq \frac{1}{\osandwich{v}{\frameop_{\mathrm{LC}}}{v}} = 3^{\abs{\suppLC(v)}}$, where $\frameop_{\mathrm{LC}}$ is the frame operator associated with \acp{LC} shadows (see \cref{sec:frame_operator_LC} for the derivation of its matrix elements).
Notice also that, although this bound corresponds to the \emph{shadow norm} \cite{huangPredictingManyProperties2020}, one does not need to maximize over all the states. 
The inequality solely originates in disregarding the square of the first-moment $(\EE_{i,U} f_{W(v)}(i,U))^2$ which agrees by construction with $\Tr(W(v)\rho)^2$ for any ensemble.
Therefore, we are most of the time comparing the exact expressions for the second moment $\EE_{i,U}[f_{W(v)}(i,U)^2]$ allowing us to formally deduce lower and upper bounds.
In the following, all our expressions for the variances are understood as being up to first-moment terms and we write, e.g.~$\varLC = 3^n$.

As derived in the previous section, the variance for the brickwork circuit depends on the partitioning of the brick support into local factors.
In principle, for any Pauli string $v$, we are able to compute such variance by means of \cref{thm:main}.
We also remark that the value obtained in this way are strict upper bounds, since we are only disregarding the square of the first moment term, which is upper bounded by $1$.
For simplicity, to compare the \ac{BW} and \acp{LC} ensembles, we derive lower and upper bounds to the exact variance expression that make the asymptotic scaling transparent.

We obtain the simplest expression when
$\partBW(v) = (n/2)$.
In this case, Theorem~\ref{thm:main} together with Proposition~\ref{prop:variance_upper_bound}
implies that
$0.8\cdot2^{n} < \varBW < 2.1^{n+1}$; see \cref{app:bounds} for details.

To compare the scaling of $\varBW$ to the one of $\varLC$, we introduce some notation to distinguish different regimes.
First, recall that we say that $v \in \FF_2^{2n}$ is supported on the $i$th brick if $\tilde v_i = 1$, and, by definition, $\tilde v_i = 1$ if at least one of $\wt(v)_{2i}$ and $\wt(v)_{2i+1}$ is non-zero.
A supported brick can further be of two types.
If $\wt(v)_{2i} \wedge \wt(v)_{2i+1} = 1$, namely the \emph{logical and} between the two local weights is non-trivial, the $i$th brick is said to be \emph{fully supported}.
Otherwise, the $i$th brick is said to be \emph{half supported} if the \emph{logical xor} between the two local weights is non-trivial, or more formally $\wt(v)_{2i} \veebar \wt(v)_{2i+1} = 1$.

Still assuming $\partBW(v) = (n/2)$, we have two extreme cases:
\begin{itemize}
	\item If each brick is \emph{fully supported}, then $\varLC = 3^n > 2.1^{n+1} > \varBW$ for all $n\geq 2$.
	Thus, the brickwork circuits have an improved sample complexity compared to single qubit random Clifford unitaries.
	The number of samples is reduced by one order of magnitude for $n \geq 8$ and by a factor of about $0.5\cdot 10^{-4}$ for $n = 20$.
	\item If each brick is \emph{half-supported},
	then $\varLC = 3^{n/2} < 1.75^n < 0.8\cdot 2^n < \varBW$ for all $n\geq 2$. 
	In this case, the \ac{BW} circuit retains its periodic structure, while \acp{LC} shadow sees the `correct' number of qubits in the support leading to a smaller sample complexity.
\end{itemize}
Similar considerations apply if $\partBW(v) = (n/2-1)$ i.e.\ when 
$\osandwich v S v$ is given by a single term with open boundary condition. 
Evaluations of the expressions for both cases are summarized in Figure~\ref{fig:BW_vs_LC}.
We observe that the scaling for both cases only differ in a constant factor as we also explain analytically in \cref{app:bounds}.

The two extreme cases suggest that shadows with \ac{BW} circuits outperform the \acp{LC} ones when the number of fully supported bricks reaches a certain threshold.
More specifically, 
based on our bounds, we can guarantee a lower sample complexity with \ac{BW} circuits 
if $\abs{\suppLC(v)} \geq 0.68 (n + 1)$ for the cases $\partBW(v) = (n/2)$ and $\partBW(v) = (n/2-1)$, see \cref{app:bounds}.
Furthermore, the additional constant term can be decreased for larger number of qubits.
Evaluations of the threshold for up to $100$ qubits are summarized in Figure~\ref{fig:thresholds_BW}.

\begin{figure}
	\centering
	\begin{subfigure}[ht]{.5\textwidth}
		\includegraphics[width=1\linewidth]{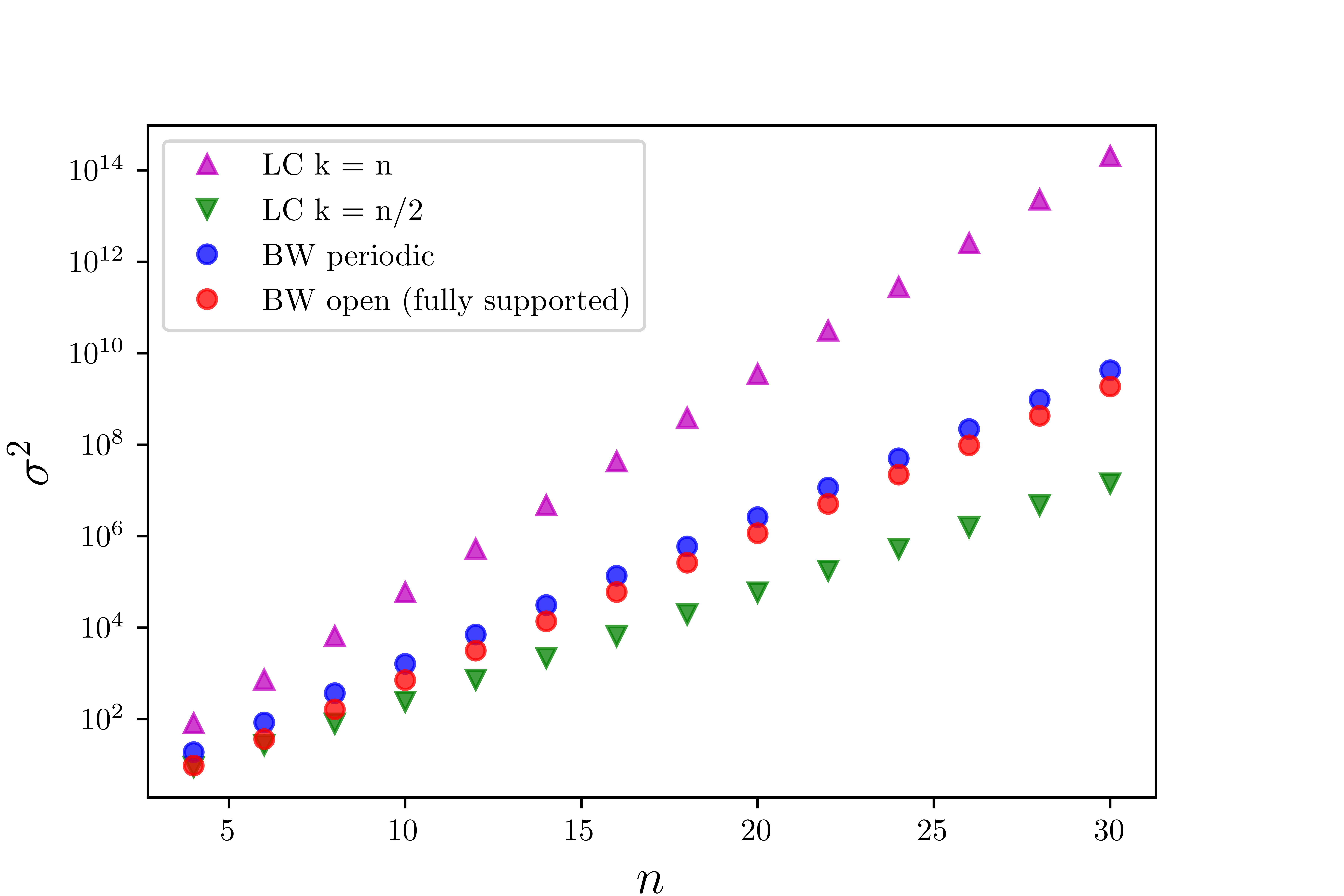}
		\caption{}
		\label{fig:BW_vs_LC}	
	\end{subfigure}%
	\begin{subfigure}[ht]{.5\textwidth}
		\includegraphics[width=1\linewidth]{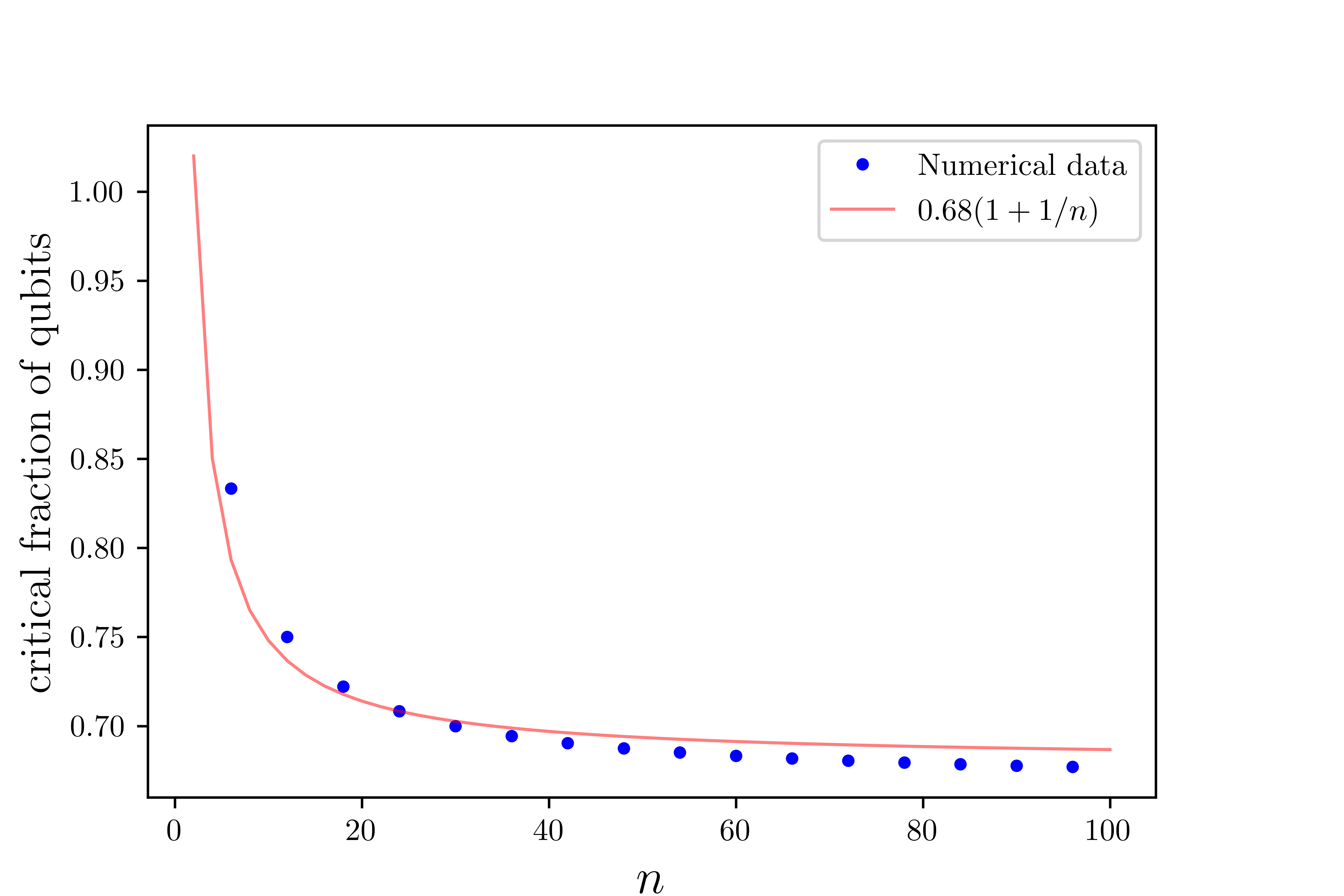}
		\caption{}
		\label{fig:thresholds_BW}	
	\end{subfigure}
	\caption{
		Figure~\ref{fig:BW_vs_LC}: 
		Comparison of variances (exact evaluation of second moments) calculated for brickworks and local Cliffords.
		For magenta and green triangles, we assumed $\suppLC(v) = n$ and $\suppLC(v) = n/2$, respectively.
		For blue and red dots, we assumed $\partBW(v) = (n/2)$ and $\partBW(v) = (n/2-1)$, respectively.
		Figure~\ref{fig:thresholds_BW}: 
		Numerical evaluation of the threshold that determines the \ac{BW} sample advantage over \acp{LC} shadows in the case $\partBW(v) = (n/2)$.
		On the $x$ axis, we represent the total number of qubits.
		On the $y$ axis, we represent the ratio between the smallest number of qubits such that $\varBW \leq \varLC$ and the total number of qubits.
		The red line represents the ratio between the lower bound for $\abs{\suppLC(v)}$ and the total number of qubits.
		The numerical dots are obtained by a direct comparison of $\varBW$ and $\varLC$:
		For each fixed $n$, and starting from the case where each brick is half-supported, we evaluated both of them for an increasing number of qubits (i.e.\ increasing the number of fully supported bricks), until the condition $\varBW \leq \varLC$ has been satisfied. 
	}
\end{figure}

Relaxing the restriction on $\partBW$, we can ensure that $\varBW \leq \varLC$ provided that
\begin{equation}
	\label{eq:threshold}
	\abs{\suppLC(v)} \geq 0.8 |\partBW(v)|  + 1.4\, 
	|\tilde v| 
	\, ,
\end{equation}
where $|\partBW(v)|$ is the number of entries in $\partBW(v)$ and $|\tilde v|$ the Hamming weight of $\tilde v$, i.e.\ the number of supported bricks in the circuit.
The derivation of this criterion is given in \cref{app:bounds}.

The main contribution (up to rescaling factors) is given by the size of $\suppBW$ (which, in turn, also influences the number of entries in $\partBW$), while the number of connected components in the effective circuit can be seen as a `correction' to the naive comparison between the notions of supported Pauli's and bricks.
In fact, by \cref{eq:vSv:general_main}, the criterion is more likely to be satisfied if the local Pauli's are bunched together: 
Sparse Pauli observables are associated with (effective) \ac{BW} circuits with many disjoint partitions, which imply a higher threshold.
For instance, for a fixed $|\suppLC(v)|$, as the the number of distinct partitions increases, each subcircuit is less densely populated, and the threshold for each subcircuit becomes harder to reach.
\cref{fig:threshold_examples} shows two non-fully supported Pauli observables, supported on a different number of qubits, associated with the \ac{BW} circuit structure.
\cref{eq:threshold} is satisfied by the second Pauli string only.
\begin{figure}[h]
	\centering
	\includegraphics[width=0.8\textwidth]{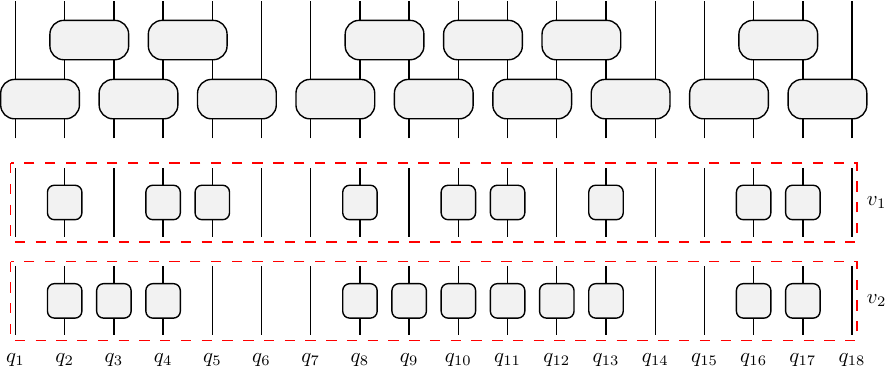}
	\caption{
		Effective circuits associated with two non-fully supported Pauli observables $v_1, v_2$.
		\cref{eq:threshold} is satisfied by $v_2$ only. 
		In particular, we have $|\partBW(v_1)| = |\partBW(v_2)| = (2, 3, 1), |\suppBW(v_1)| = 9, |\suppBW(v_2)| = 11$, which imply $\varBW(v_1) = \varBW(v_2) \approx 58 \cdot 10^3, \varLC(v_1) \approx 19 \cdot 10^3, \varLC(v_2) \approx 177 \cdot 10^3$.
	}
	\label{fig:threshold_examples}
\end{figure}

The threshold criterion is likely to hold for random Pauli observables, since, for a fixed $n$, few additional qubits are needed to reach the threshold.
On the other hand, for any random Pauli string $v \in \FF_2^{2n}$, the probability of the $i$th brick to be fully supported is strictly larger than the probability of being half supported.
Indeed, evaluating the bounds for random Pauli strings we observe that the brickwork circuit gives better performance with high probability drawing random Pauli strings, and $p(\varBW \leq \varLC) \xrightarrow{n \text{\ big} } 1$, see Figure~\ref{fig:BW_vs_LC_probabilistic}.
\begin{figure}[ht]
	\centering
	\includegraphics[width=0.5\linewidth]{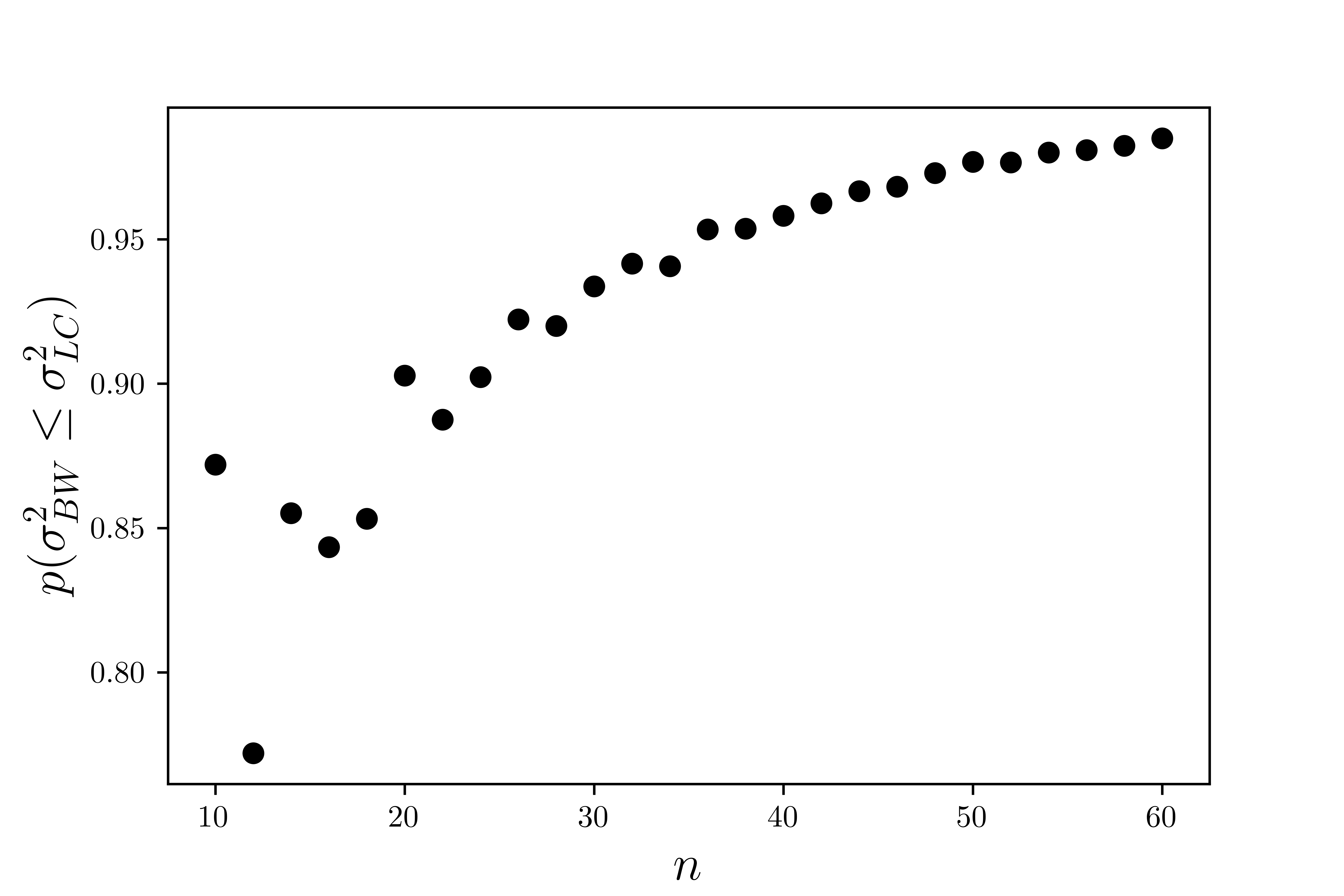}
	\caption{
		Probability of $\varBW \leq \varLC$, evaluated on $2^{16}$ random bitstrings for each value of $n$.
		}
	\label{fig:BW_vs_LC_probabilistic}	
\end{figure}

In conclusion, we showed that \acp{LC} shadows still have their own merit, in particular they are still the best choice for very sparse, local Pauli observables.
However, they are significantly outperformed in all the other cases.
For instance, for fully supported bricks, we observed that the variance is scaling as $\approx 2.1^{n}$, which is very close to the performance of a global Cliffords ensemble for a moderate number of qubits, namely global properties may be predicted using brickwork shadows.
This represents a particular case of the shallow shadows presented in \cite{bertoniShallowShadowsExpectation2022}, where the authors argue that brickwork circuits (in their case, of depth $\log(n)$) are expected to have the same sample complexity as the global Clifford scheme.

%%%%%%%%%%%%%%%%%%
\subsection{Numerical experiments}
\label{sec:numerics}
%%%%%%%%%%%%%%%%%%

We now compare numerically performances of \ac{BW} and \acp{LC} estimation procedures.
We fix $n = 10$ as the number of qubits, and consider for simplicity $\rho = \ketbra{0}{0}$ as the input state.
Then, we collect numerical data for three different $Z$-type operators, that we assume to be supported on each brick.
In particular, we consider the following Pauli strings:
$\vfull$, which is supported on each qubit, $\vhalf$, where each brick is half supported, and $\vthres$, which is supported on $8$ qubits, ensuring it satisfy the threshold criterion discussed in the last section.
Notice that it does not matter where the two half supported bricks are located in $\vthres$, since all of them are supported.
Finally, drawing unitaries from the Clifford group, we can classically simulate the whole procedure efficiently using standard techniques \cite{gottesmanHeisenbergRepresentationQuantum1998, aaronsonImprovedSimulationStabilizer2004,koenig_how_2014}.
More details on the algorithms are provided in \cref{app:numerics} and at the following link: \url{https://github.com/MirkoArienzo/shadow_short_circuits}.

Then, we fix $m$ as the number of samples, and compute the empirical average over all samples as described in \cref{sec:classical_shadows}, which yields an estimator for the given observables and $\rho = \ketbra{0}{0}$. 
We run this procedure $100$ times, and evaluate the average of the estimators over all runs.
The latter has standard deviation given by $\sigma / \sqrt{100 m}$, with $\sigma = \stdevBW, \stdevLC$.
Finally, the task is repeated for different values of $m$.

The results of the simulations, shown in Figure~\ref{fig:numerics}, agree with the previous discussion.
In particular, for circuits that are fully supported or over the threshold, the convergence to the expected values is faster using \ac{BW} circuits, see Figure~\ref{fig:numerics_BW_vs_LC_full} and Figure~\ref{fig:numerics_BW_vs_LC_threshold}, while the converse happens in the case of half supported circuits, see Figure~\ref{fig:numerics_BW_vs_LC_half}.
\begin{figure}
	\centering
	\begin{subfigure}[ht]{.5\textwidth}
		\includegraphics[width=1\linewidth]{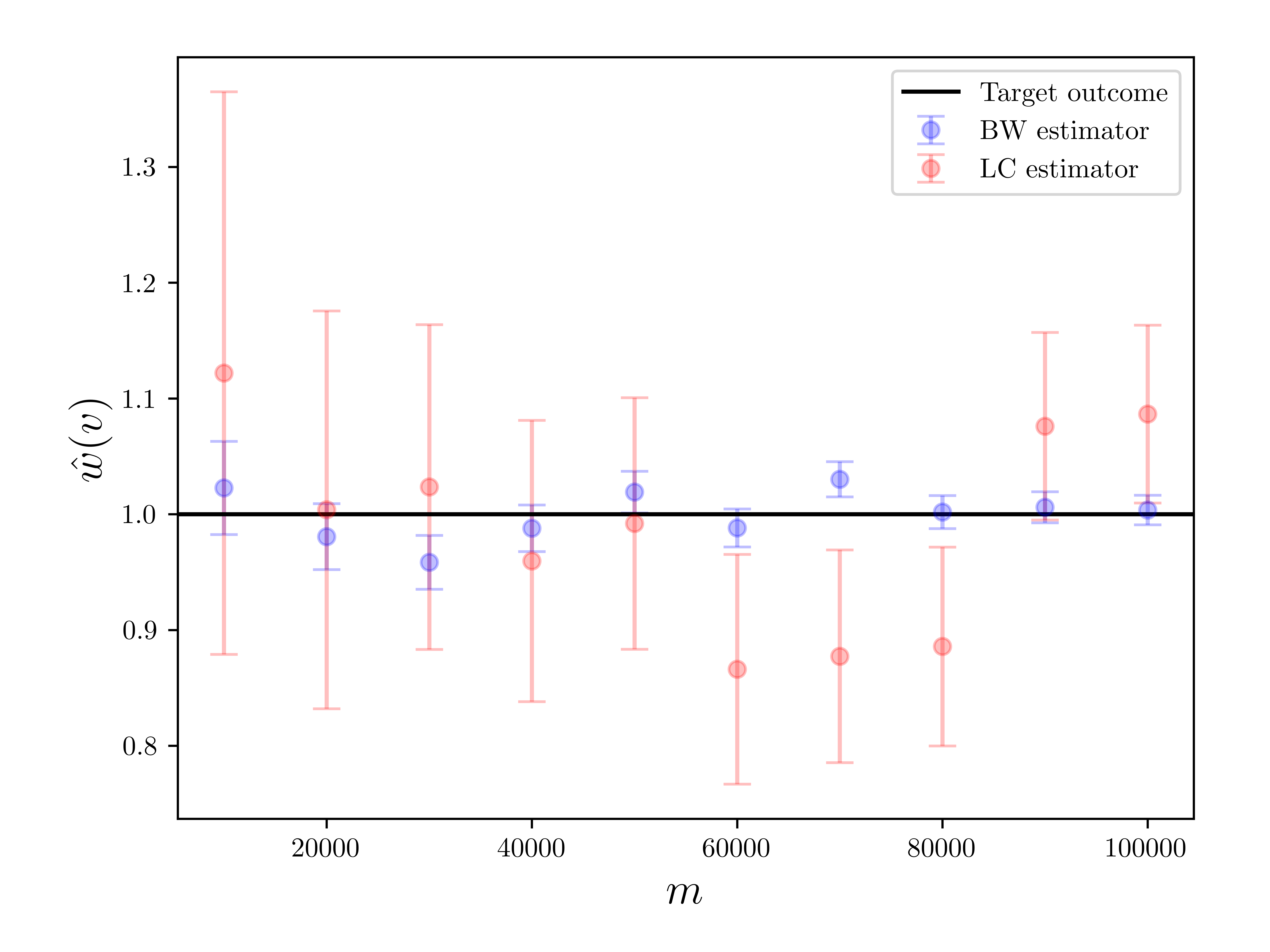}
		\caption{$\vfull$ plot: $10$ qubits in the support of $v$.}
		\label{fig:numerics_BW_vs_LC_full}	
	\end{subfigure}%	
	\begin{subfigure}[ht]{.5\textwidth}
		\includegraphics[width=1\linewidth]{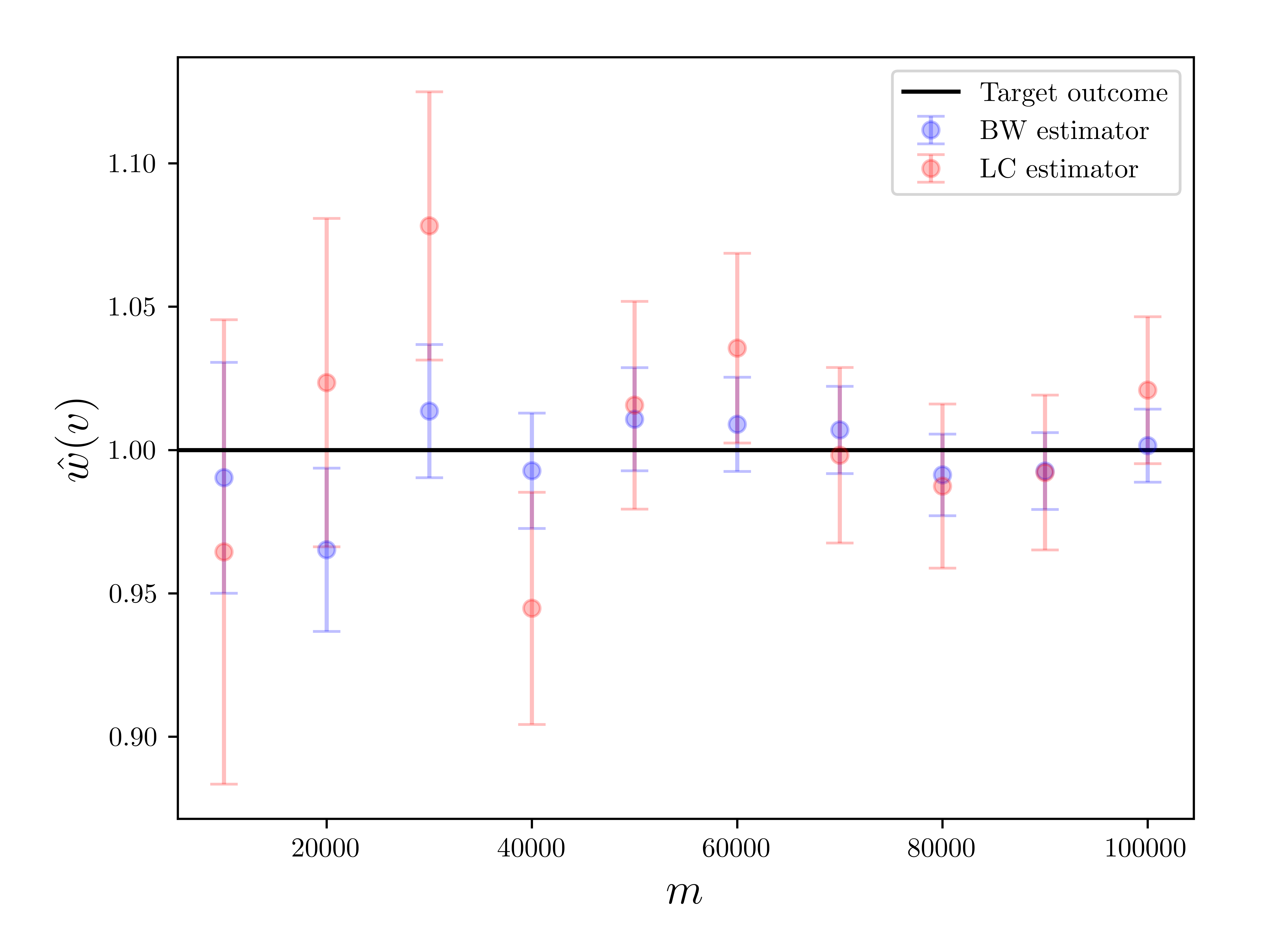}
		\caption{$\vthres$ plot: $8$ qubits in the support of $v$.}
		\label{fig:numerics_BW_vs_LC_threshold}	
	\end{subfigure}%	
	\\
	\begin{subfigure}[ht]{.5\textwidth}
		\includegraphics[width=1\linewidth]{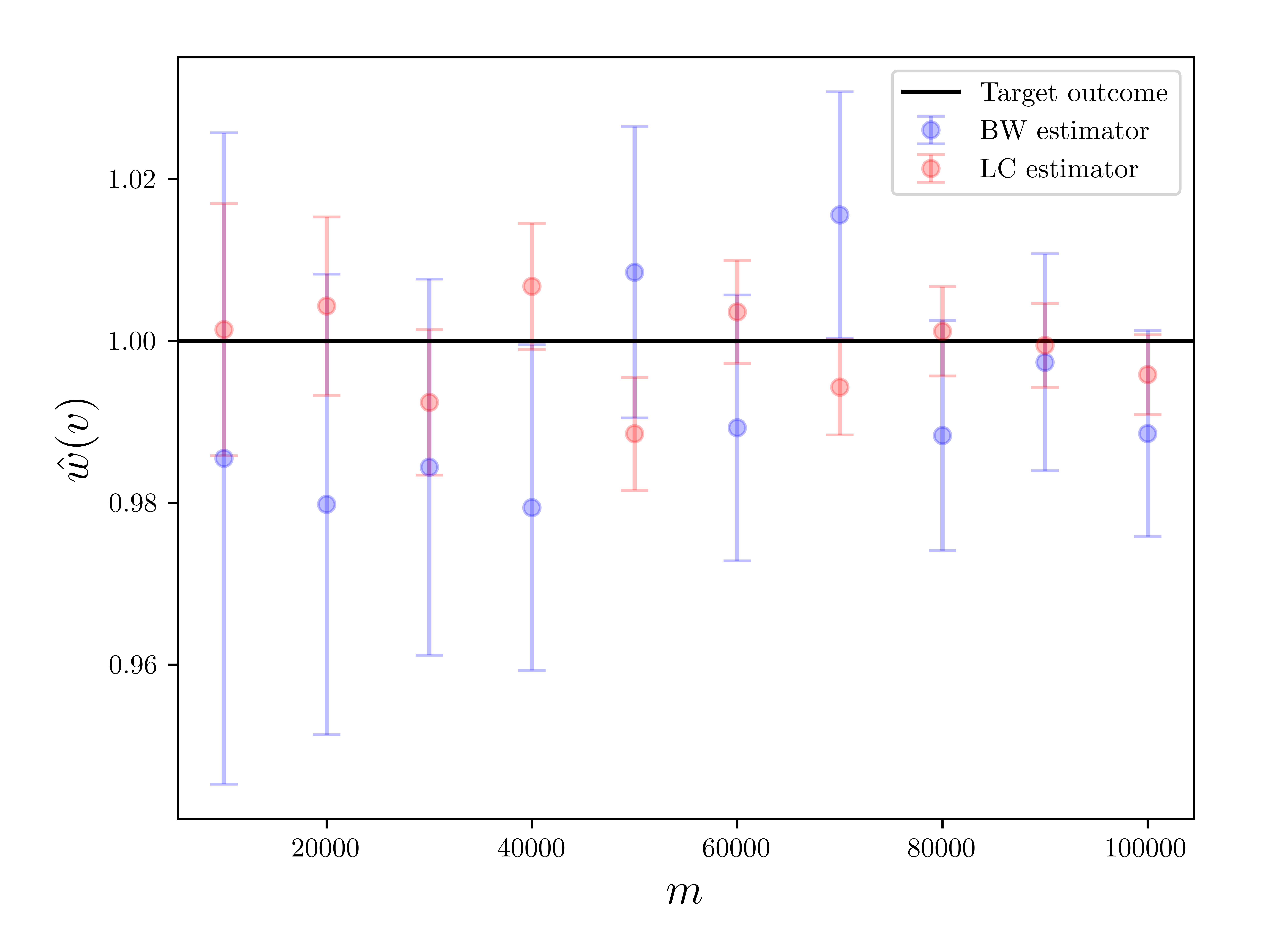}
		\caption{$\vhalf$ plot: $5$ qubits in the support of $v$.}
		\label{fig:numerics_BW_vs_LC_half}	
	\end{subfigure}%	
	\caption{
		Convergence of the estimators $\hat w(\vfull)$, $\hat w(\vthres)$, and $\hat w(\vhalf)$, respectively, as defined in \cref{sec:classical_shadows}.
		We consider a system of $10$ qubits with input state $\rho=\ketbra{0}{0}$, ensuring we can classically simulate the whole procedure efficiently.
		For each fixed $m$, $100$ runs have been performed and then the average over all of them, with the respective standard deviation, has been plotted.
		BW estimator is converging faster for $\vfull$ and $\vthres$.
	}
	\label{fig:numerics}
\end{figure}

%%%%%%%%%%%%%%%%%%
\section{Conclusions}
%%%%%%%%%%%%%%%%%%
Shadow estimation with short circuits can interpolate between the originally proposed constructions with local and global Clifford unitaries.
We derived closed-form analytic expressions for the frame operator (and its inverse) associated with the arguably simplest circuit construction: one round of a brickwork circuit.
In particular, we observed how the $2$-design property of bricks can be used to determine systems of recurrence relations for the contributions of subcircuits with effectively periodic and open boundary conditions.
The recurrence relations admit closed-form solutions and can be used to calculate the classical shadow and the corresponding linear estimators.
Furthermore, and in contrast to numeric approaches, we explicitly worked out and analyse the sample complexity of Pauli estimation with one round of a brickwork circuit.
This gave rise to a simple criterion for the structure of the support of the Pauli observable in order to have a scaling advantage compared to using local Clifford unitaries.
Going beyond the worked out example, our results 
provide clear evidence for the potential of using short circuits for shadow estimation but also indicate limitations and the need for careful comparison in specific applications.
Besides shadow estimation, the analytic expression for the frame operator can potentially also unlock the usage of short depth circuits in related tasks involving randomized measurements such as benchmarking and mitigation.

We expect that generalizing our analytic approach will become considerably more involved for deeper circuits, especially for deriving the exact frame operator required for constructing the classic shadow.
To this end, we regard numerical methods as a considerably more flexible approach.
Nonetheless, it might be possible to derive sample complexity bounds for deeper circuits and other estimation tasks following and generalizing the argument presented here.
Altogether, our work provides both the analytical results and the motivation for implementing short depth quantum circuits in an actual shadow estimation experiment.
After all, the merits of the approach have to be evaluated in practice.

%%% ---------------------------------------------
\section{Acknowledgements}
%%% ---------------------------------------------
We thank Leandro Aolita and Renato Mello for fruitful discussions on shadow estimation and 
Christian Bertoni, Jonas Haferkamp, Marcel Hinsche, Marios Ioannou, Jens Eisert, and Hakop Pashayan for discussions on relations to their work \cite{bertoniShallowShadowsExpectation2022}. 
This work has been funded by 
the Deutsche Forschungsgemeinschaft (DFG, German Research Foundation) within the Emmy Noether program (grant number 441423094), 
the German Federal Ministry of Education and Research (BMBF) within the funding program ``quantum technologies -- from basic research to market'' via the joint projects
MIQRO (grant number 13N15522)
and 
MANIQU (grant number 13N15578), 
and by 
the Fujitsu Services GmbH as part of the endowed professorship ``Quantum Inspired and Quantum Optimization''.

%%%%%%%%%%%%%%%%%%%%%%%%%
\newpage
\section*{Appendices}
\appendix
%%%%%%%%%%%%%%%%%%%%%%%%%

%%%%%%%%%%%%%%%%%%
\section{Auxiliary results}
\label{app:auxiliary_results}
%%%%%%%%%%%%%%%%%%

Here, we include some basic facts concerning moment operators.
For a self-contained introduction, see \cite[Section~II.I]{Kliesch2020TheoryOfQuantum}.

Let $\mu$ be the Haar probability measure on the unitary group $\U(\gdim)$.
The $k$-\emph{th moment operator} of $\mu$ is defined by
\begin{equation}
	\mathcal M_\mu^{(k)} (A)
	\coloneqq
	\int_{\U(\gdim)} d\mu(U) \, U^{\otimes k} A U^{\ad \otimes k}, \quad A \in \CC^{k\gdim \times k\gdim}. 
\end{equation}
The operator $\mathcal M_\mu^{(k)}(A)$ is often referred to as the group twirling of the operator $A$.
It is very easy to check with the invariance of the Haar measure that $\mathcal M_\mu^{(k)}$ commutes with $U^{\otimes k}$ for any $U \in \U(\gdim)$, and in particular it is the projector onto the set of such operators endowed with the Hilbert-Schmidt inner product.
Hence, by Schur-Weyl duality, one can prove that $\mathcal M_\mu^{(k)}$ can be expressed as the linear combination of some suitable projectors associated with the irreps of the symmetric group $\Sym_k$.
In particular, we will use the following result:
\begin{lemma}
	\label{lem:Haar_second_moment}
	Let $\nu$ be the normalized Haar measure on $\U(\gdim)$.
	Then, for any integer $k \geq 1$
	\begin{equation}
		\label{eq:k_moment_zero}
		\EE_{U \sim \nu} \,  U^{\otimes k} \ketbra 00 U^{\otimes k\dagger} 
		=
		\binom{k+\gdim-1}{\gdim-1}^{-1}
		\,  \Psym[k] \, ,
	\end{equation}
	which evaluates to $\frac{2}{\gdim(\gdim+1)} \, (\1+\flip)$ for $k=2$, where $\Psym[k]$ is the projector onto the completely symmetric subspace and $\flip\in \L(\CC^d\otimes \CC^d)$ the flip operator. 
	Moreover, given $u, v \in \FF_2^{2n}$,
	\begin{equation}
		\EE_{U \sim \nu} \,  U^{\otimes 2} W(u) \otimes W(v) U^{\otimes 2\dagger} 
		= 
		\begin{cases}
			\1 \, , \quad \text{if} \ u = 0 \ \text{and} \ v = 0 \, \\
			\frac{\delta_{u, v}}{\gdim^2-1}\, (\gdim\, \flip - \1) \, , \quad \text{otherwise}
		\end{cases} \, ,
	\end{equation} 
	where $\delta_{u,v}$ is the Kronecker $\delta$ over $\FF_2^{2n}$.
\end{lemma}
\begin{proof}
	See e.g.\ \cite[Section~II.I]{Kliesch2020TheoryOfQuantum}. 
\end{proof}
Finally, recall that a \emph{unitary $k$-design} is a probability measure $\nu$ on $\U(\gdim)$ which reproduce expectation values of the Haar measure $\mu$ up to degree $k$.
Formally, this means that, for any $A \in \C^{kN \otimes kN}$,
\begin{equation}
	\int_{\U(\gdim)} d\nu(U) U^{\otimes k} A U^{\ad \otimes k} = 
	\int_{\U(\gdim)} d\mu(U) U^{\otimes k} A U^{\ad \otimes k} \, .
\end{equation}
Therefore, the relevant results (for our purposes) in \cref{lem:Haar_second_moment} hold for any unitary $2$-design.

We remark that unitary $k$-designs are fundamental for practical implementations, since we are often interested only in the first $k$ moments and, more importantly, drawing Haar random unitaries is a very hard task.

%%%%%%%%%%%%%%%%%%
\subsection{Frame operator for Local Cliffords shadows}
\label{sec:frame_operator_LC}
%%%%%%%%%%%%%%%%%%
Here, we include the calculation for the frame operator of \ac{LC} shadows.
It follows the same steps as the calculation given in \cref{sec:proof_theorem} for \ac{BW} shadows.
However, dealing with only one layer of Clifford unitaries makes things more straightforward.
In particular, given $u, v \in \FF_2^{2n}$, and denoting by $\frameop_{\mathrm{LC}}$ the frame operator associated with \ac{LC} shadows, we have
\begin{equation}
	\label{eq:Svv_LC}
	\begin{aligned}
		\osandwich{u}{\frameop_{\mathrm{LC}}}{v} 
		& = 
		\frac{1}{\gdim}
		\sum_{i} \EE_{U \sim \mu} \obraket{u}{E_{i,U}} \obraket{E_{i,U}}{v} \\
		& =
		\frac{1}{\gdim}
		\sum_i \EE_{U \sim \mu} \Tr[W(u)^\dagger U \ketbra ii U^\dagger]\,  \Tr[U \ketbra ii U^\dagger W(v)]
		\\
		& =
		\EE_{U \sim \mu} \sandwich{0}{U^{\otimes 2\dagger} W(u) \otimes W(v) U^{\otimes 2}}{0} \, .
	\end{aligned}
\end{equation}
Considering the factorization
\begin{equation}
	W(v) = W(v_1) \otimes \dots \otimes W(v_n) \, ,
\end{equation}
by \cref{lem:Haar_second_moment}, we have
\begin{equation}
	\EE_{U \sim \mu} U^{\otimes 2\dagger} W(u) \otimes W(v) U^{\otimes 2}
	= 
	\delta_{u,v}
	\frac{1}{3^{\abs{\suppLC(v)}}} \left( 2 \, \flip - \1 \right)^{\otimes \abs{\suppLC(v)}} \, , 
\end{equation}
where $\suppLC(v) \coloneqq \{ i \mid \wt(v)_i \neq 0 \}$.

Therefore, it follows
\begin{equation}
	\begin{aligned}
		\osandwich{v}{\frameop_{\mathrm{LC}}}{v}
		& = 
		\frac{1}{3^{\abs{\suppLC(v)}}} 
		\sandwich{0}{2 \flip - \1}{0}^{\abs{\suppLC(v)}} \\
		& =
		\frac{1}{3^{\abs{\suppLC(v)}}} \, .
	\end{aligned}
\end{equation}

%%%%%%%%%%%%%%%%%%
\section{Tensor networks for Lemma~\ref{lem:recursive_relations_periodic} and \ref{lem:recursive_relations_open}}
\label{app:tensor_networks}
%%%%%%%%%%%%%%%%%%

In this section, we show how $t_1, t_2, t_3$ (respectively, $t_1, t_2$) appearing in \cref{lem:recursive_relations_periodic} (respectively, \cref{lem:recursive_relations_open}) can be written as a system of recurrence relations using tensor networks.

First, notice that each operator in the traces \cref{eq:tper(n)_def,eq:topen(n)_def} acts on two copies of $2$ qubits.
That means that each brick is represented by two overlapping copies, see \cref{fig:brick_two_copies}.
\begin{figure}[ht]
	\centering
	\includegraphics{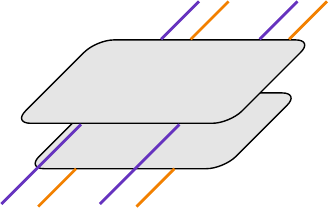}
	\caption{Each local operator corresponds to two overlapping copies of a brick in the \ac{BW} circuit.
	}
	\label{fig:brick_two_copies}
\end{figure}
Next, given (the two copies of) a brick, we set for notational purpose
\begin{equation*}
	\begin{tikzpicture}
		[scale=.8,baseline={([yshift=-.5ex]current bounding box.center)},vertex/.style={anchor=base,
			circle,fill=black!25,minimum size=18pt,inner sep=2pt}]
		\node[inner sep=0pt] (full)
		{\includegraphics{standalone/brick3D_2copies}};
	\end{tikzpicture}
	\equiv
	\begin{tikzpicture}
		[scale=.8,baseline={([yshift=-.5ex]current bounding box.center)},vertex/.style={anchor=base,
			circle,fill=black!25,minimum size=18pt,inner sep=2pt}]
		\node[inner sep=0pt] (full)
		{\includegraphics{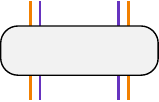}};
	\end{tikzpicture} \, .
\end{equation*}
Now notice that each brick is made up of identities and flips, the latter is usually represented as
\begin{equation}
	\flip = 
	\begin{tikzpicture}
		[scale=.8,baseline={([yshift=-.5ex]current bounding box.center)}]
		\node[inner sep=0pt] (full)
		{\includegraphics{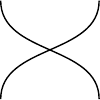}};
	\end{tikzpicture} \, .
\end{equation}
Finally, $\tper(n)$ and $\topen(n)$ will be simplified exploiting linearity and separability of bricks.
For this purpose, and to simplify the notation, we rewrite each pair of lines corresponding to the same qubit as a single one.
In particular, if two lines are straight (the identity operator is applied), we summarize them as a single black line, otherwise as a red line when the flip operator is applied.
For instance, for a brick $\1 + \flip$ in the first layer of the circuit (here $\flip \equiv \flip_{(2)}$ for simplicity), we have
\begin{equation*}
	\begin{tikzpicture}
		\node[inner sep=0pt] (full)
		{\includegraphics{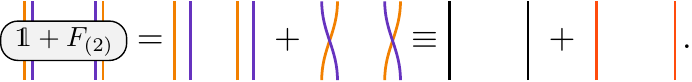}};
	\end{tikzpicture}
\end{equation*}

\subsection{Proof of relations (\ref{eq:t1_t2_t3_periodic_recursive}) and (\ref{eq:t1_t2_t3_periodic_base})}
We show how to derive recurrence relations in \cref{lem:recursive_relations_periodic}.
First, recall the following definitions:
\begin{align}
	t_1(n)
	& \coloneqq
	\Tr \left[ 
	\left(
	\1 \otimes \left( \1 + \flip_{(2)} \right)^{\otimes n/2-1} \otimes \1
	\right)
	\left(
	4 \flip_{(2)} -\1
	\right)^{\otimes n/2} 
	\right] \\
	& = 
	\begin{tikzpicture}
		[scale=.8,baseline={([yshift=-.5ex]current bounding box.center)}]
		\node[inner sep=0pt] (full)
		{\includegraphics{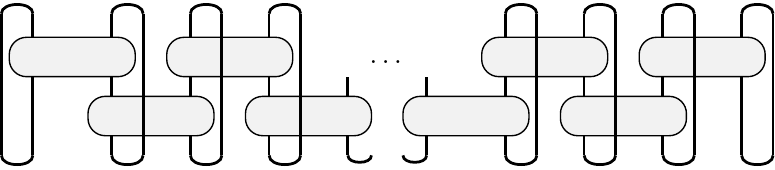}};
	\end{tikzpicture} \, , \\
	t_2(n)
	& \coloneqq
	\Tr \left[ 
	\left(
	\flip \otimes \left( \1 + \flip_{(2)} \right)^{\otimes n/2-1} \otimes \flip
	\right)
	\left(
	4 \flip_{(2)} -\1
	\right)^{\otimes n/2} 
	\right] \\
	& = 
	\begin{tikzpicture}
		[scale=.8,baseline={([yshift=-.5ex]current bounding box.center)}]
		\node[inner sep=0pt] (full)
		{\includegraphics{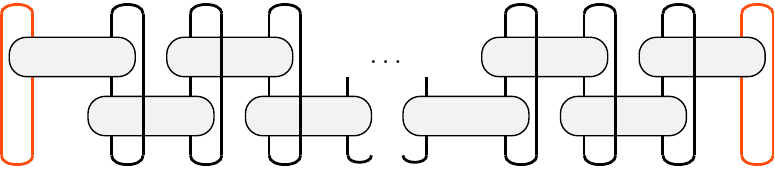}};
	\end{tikzpicture} \, , \\
	t_3(n)
	& \coloneqq
	\Tr \left[ 
	\left(
	\1 \otimes \left( \1 + \flip_{(2)} \right)^{\otimes n/2-1} \otimes \flip
	\right)
	\left(
	4 \flip_{(2)} - \1
	\right)^{\otimes n/2} 
	\right] \\
	& = 
	\begin{tikzpicture}
		[scale=.8,baseline={([yshift=-.5ex]current bounding box.center)}]
		\node[inner sep=0pt] (full)
		{\includegraphics{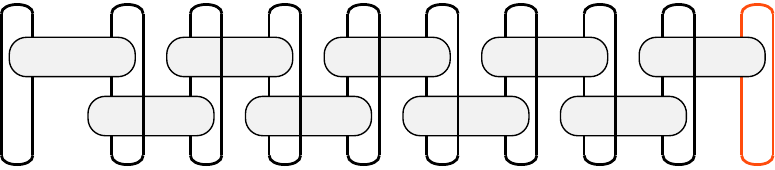}};
	\end{tikzpicture} \, .	
\end{align}
Therefore, we have
\begin{equation*}
	\begin{tikzpicture}
		[scale=.8,baseline={([yshift=-.5ex]current bounding box.center)}]
		\node[inner sep=0pt] (full)
		{\includegraphics{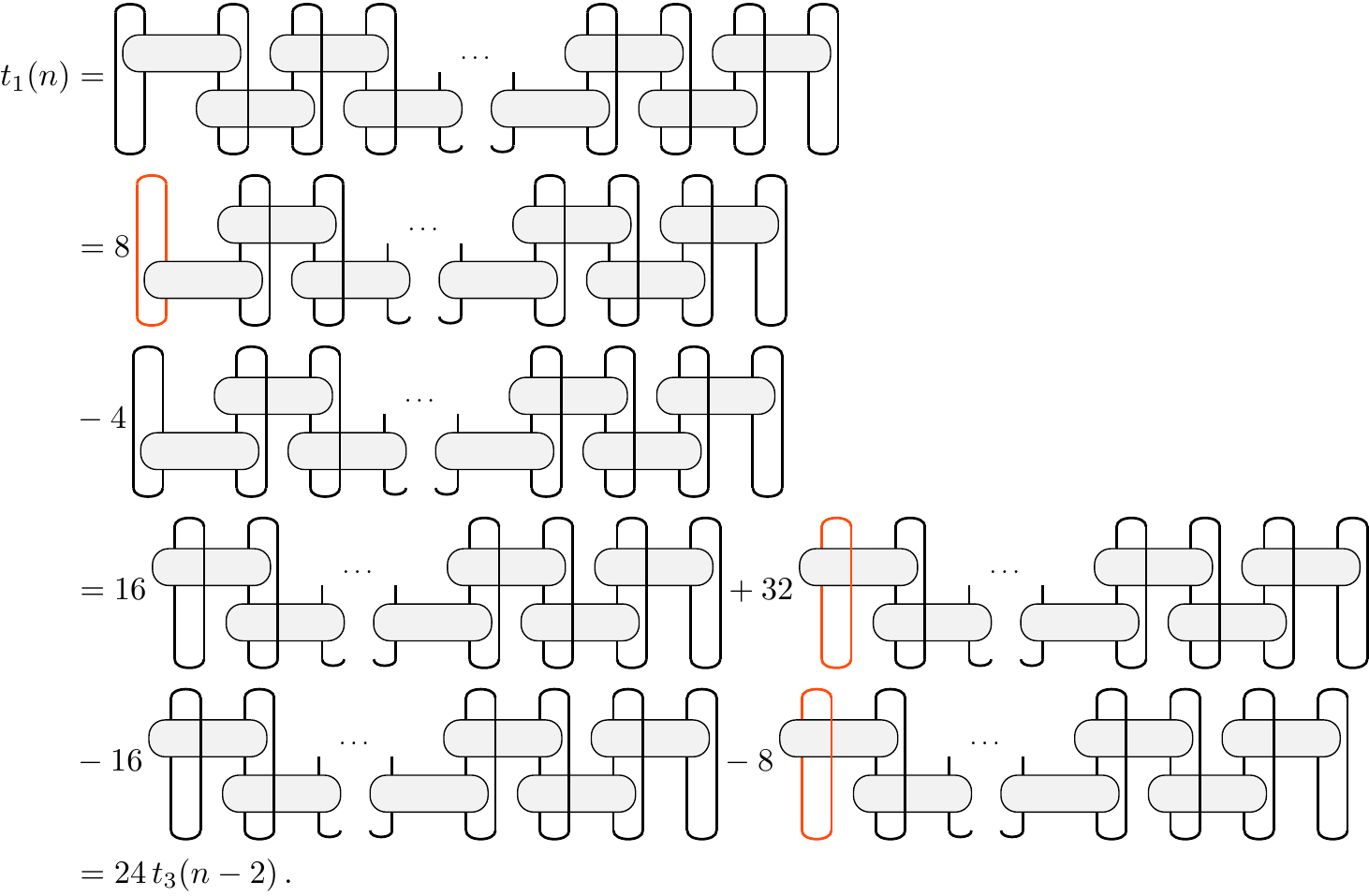}};
	\end{tikzpicture}
\end{equation*}
Similarly, for $t_2$ we obtain
\begin{equation*}
	\begin{tikzpicture}
		[scale=.8,baseline={([yshift=-.5ex]current bounding box.center)}]
		\node[inner sep=0pt] (full)
		{\includegraphics{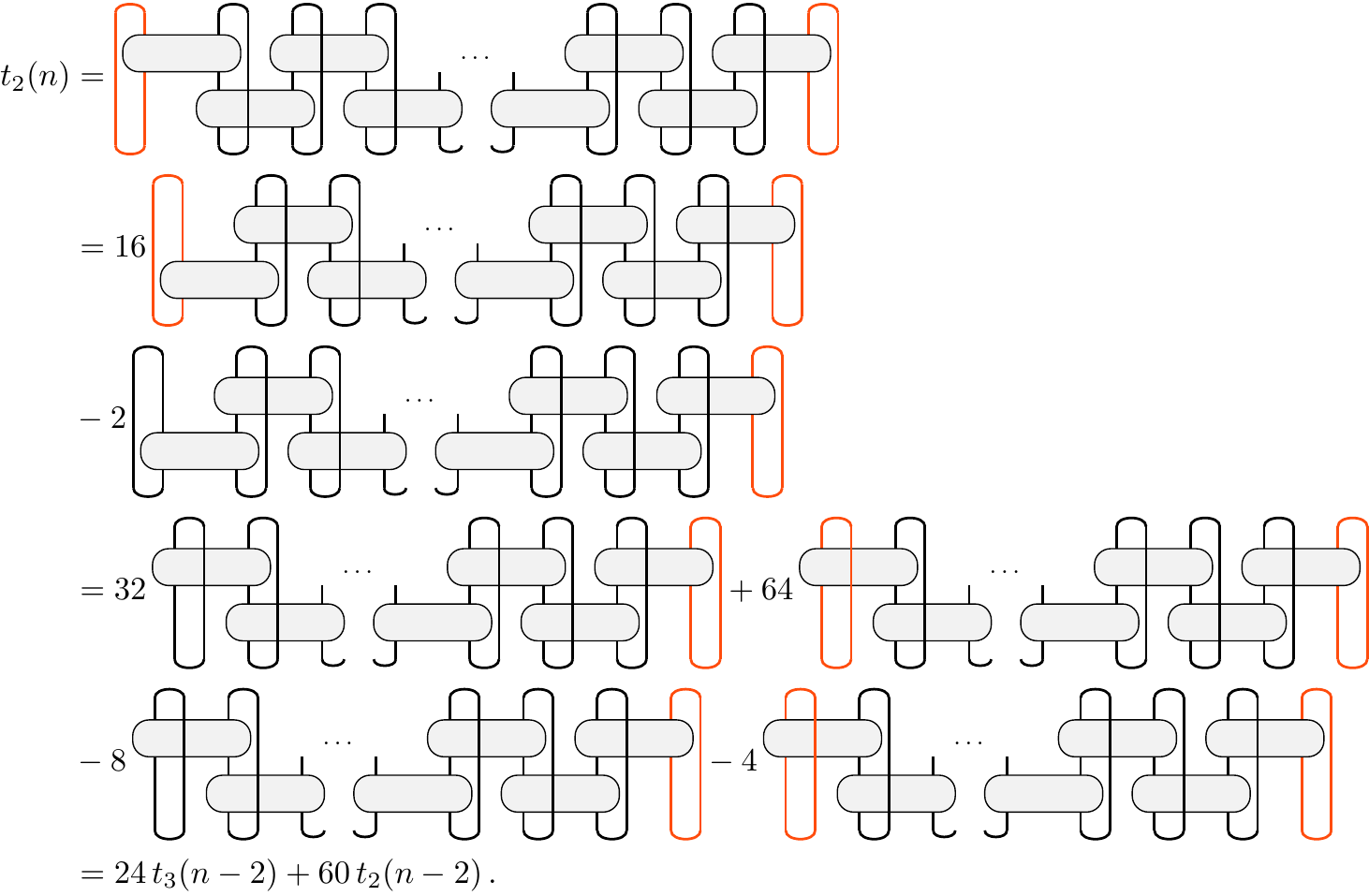}};
	\end{tikzpicture}
\end{equation*}
Finally, for $t_3$ a similar calculation yields
\begin{equation*}
	\begin{tikzpicture}
		[scale=.8,baseline={([yshift=-.5ex]current bounding box.center)}]
		\node[inner sep=0pt] (full)
		{\includegraphics{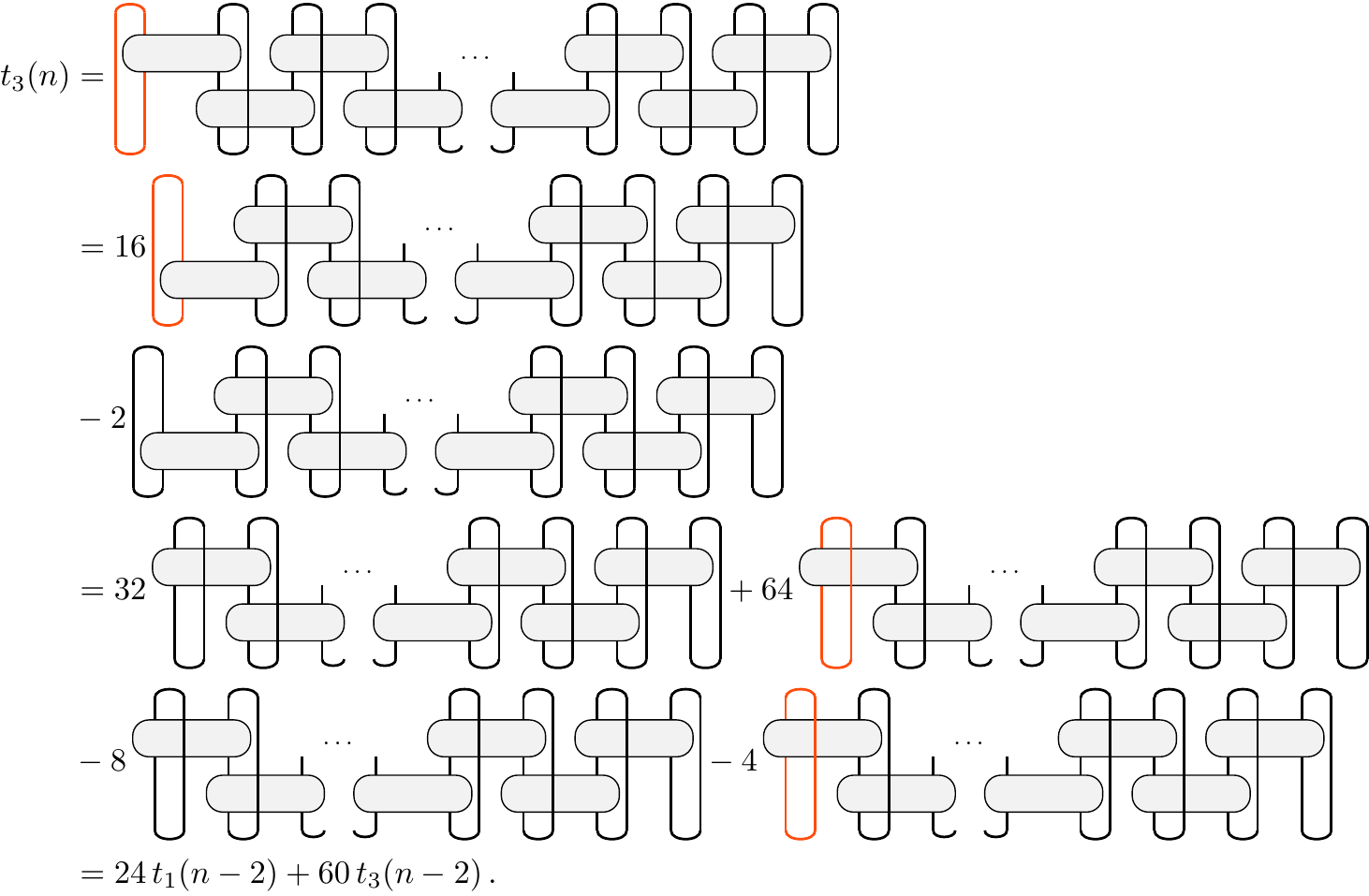}};
	\end{tikzpicture}
\end{equation*}
Moreover,
\begin{equation*}
	\begin{tikzpicture}
		[scale=.8,baseline={([yshift=-.5ex]current bounding box.center)}]
		\node[inner sep=0pt] (full)
		{\includegraphics{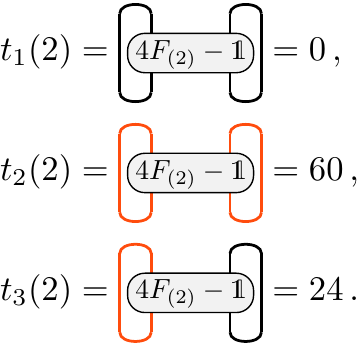}};
	\end{tikzpicture}
\end{equation*}

\subsection{Proof of relations (\ref{eq:t1_t2_open_recursive}) and (\ref{eq:t1_t2_open_base})}
Recall that, in this case,
\begin{align}
	t_1(n) 
	& \coloneqq
	\Tr\myleft[ \left\{ \1_{4}  \otimes (\1 + \flip_{(2)})^{\otimes (n/2-1)}\right\}
	\left\{ (4 \flip_{(2)} -\1)^{\otimes (n/2-1) } \otimes \1_{4} \right\}  \myright]  \\
	& = 
	\begin{tikzpicture}
		[scale=.8,baseline={([yshift=-.5ex]current bounding box.center)}]
		\node[inner sep=0pt] (full)
		{\includegraphics{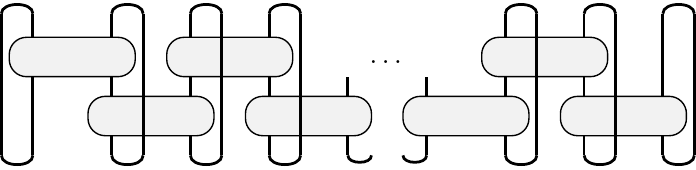}};
	\end{tikzpicture}
	\, , \\
	t_2(n) 
	& \coloneqq
	\Tr\myleft[ \left\{  \flip  \otimes (\1+\flip_{(2)})^{\otimes (n/2-1)}\right\}
	\left\{ (4 \flip_{(2)} -\1)^{\otimes (n/2-1) } \otimes \1_{4} \right\}  \myright] \\
	& = 
	\begin{tikzpicture}
		[scale=.8,baseline={([yshift=-.5ex]current bounding box.center)}]
		\node[inner sep=0pt] (full)
		{\includegraphics{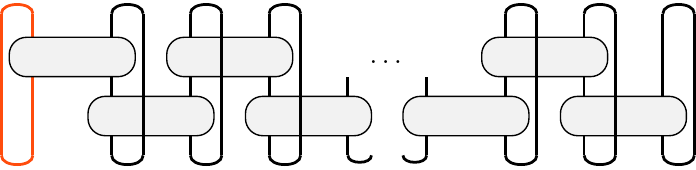}};
	\end{tikzpicture} 
	\, . 
\end{align}
For the first trace we have
\begin{equation*}
	\begin{tikzpicture}
		\node[inner sep=0pt] (full)
		{\includegraphics{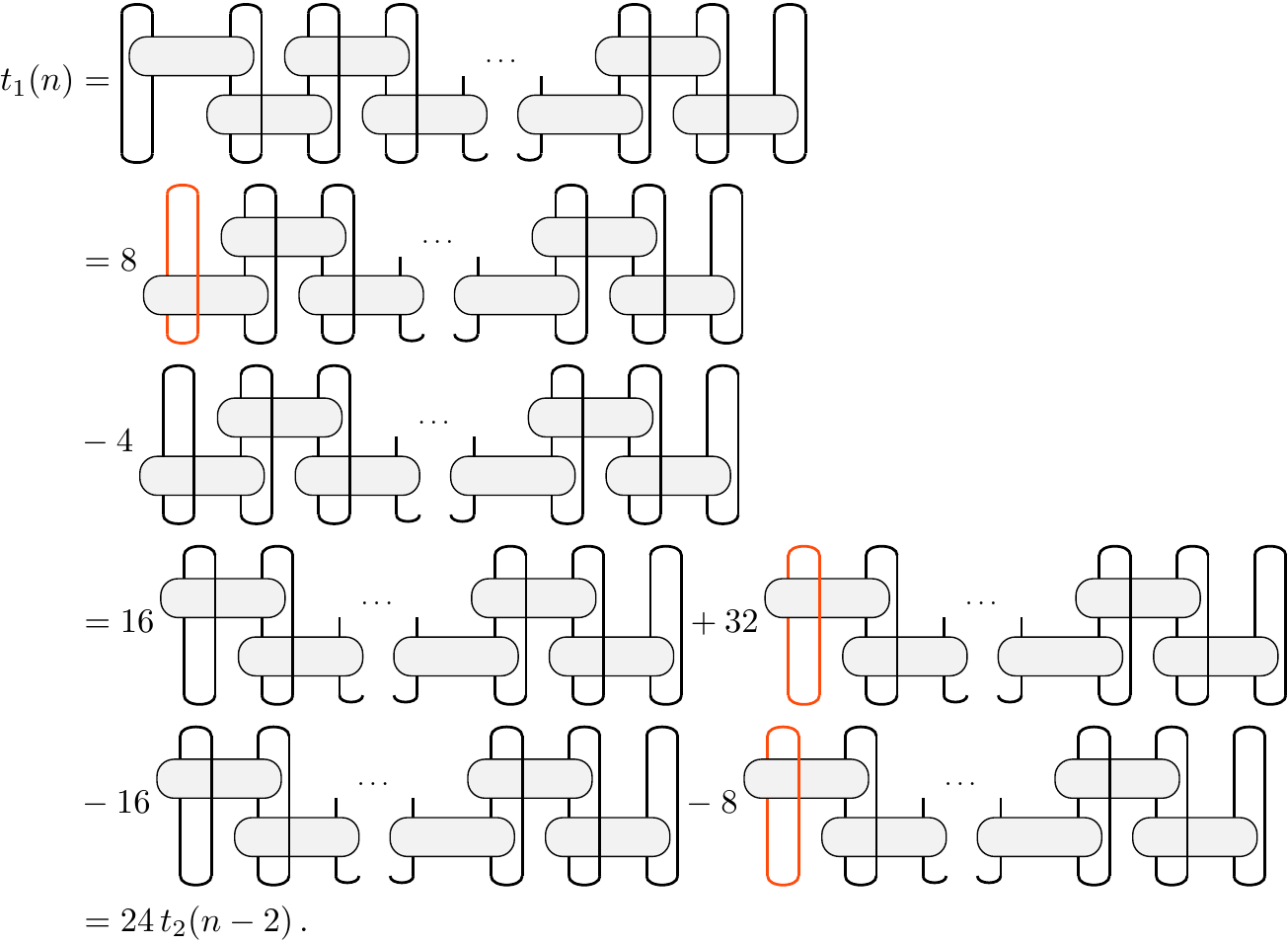}};
	\end{tikzpicture}
\end{equation*}
Similarly,
\begin{equation*}
	\begin{tikzpicture}
		\node[inner sep=0pt] (full)
		{\includegraphics{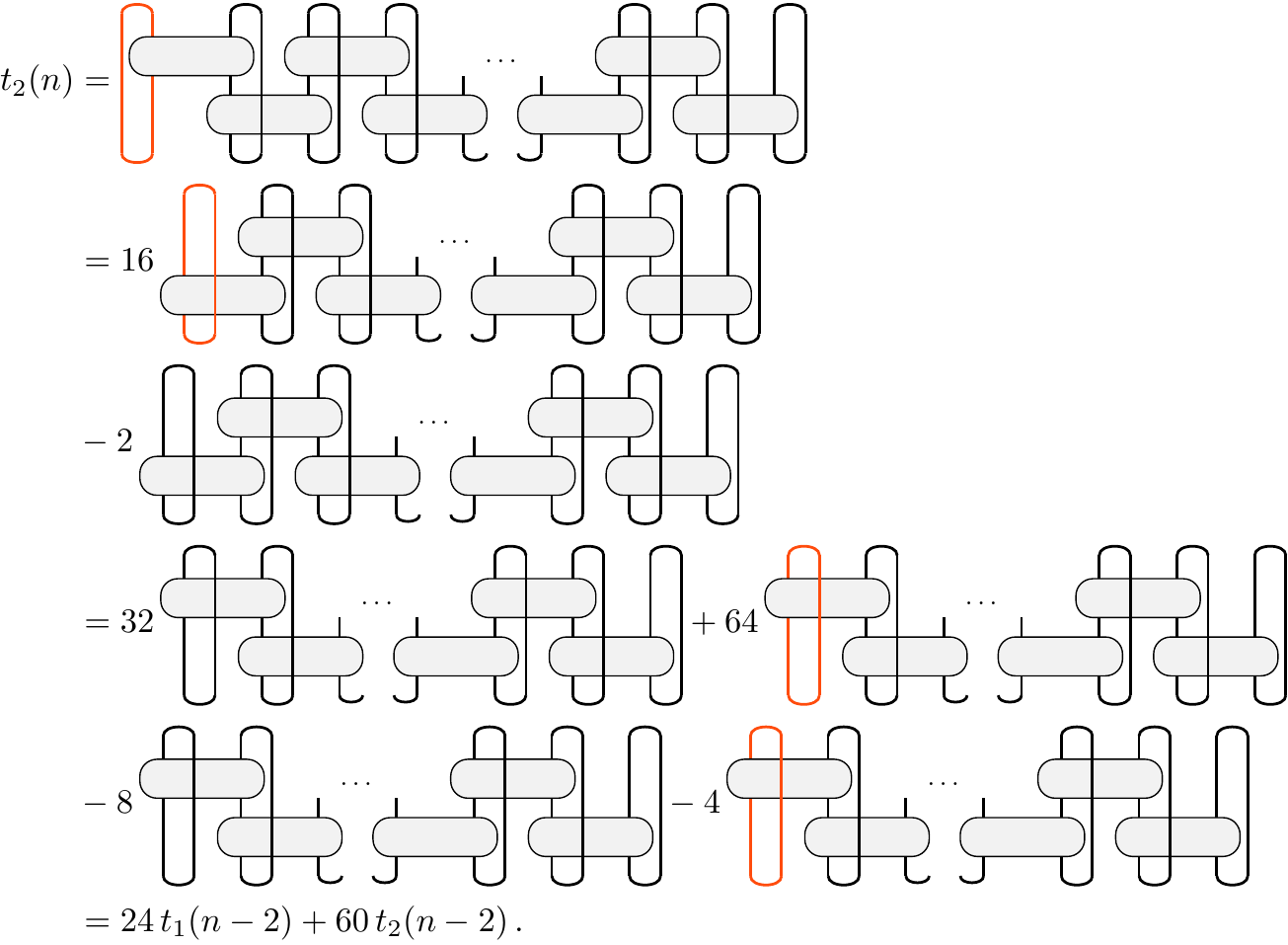}};
	\end{tikzpicture}	
\end{equation*}
Moreover,
\begin{equation*}
	\begin{tikzpicture}
		\node[inner sep=0pt] (full)
		{\includegraphics{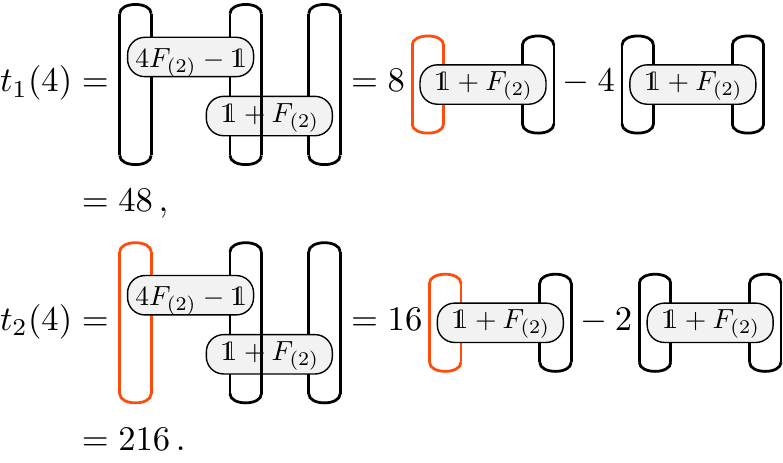}};	
	\end{tikzpicture}
\end{equation*}

%%% ===================================
\section{Another take on the variance bound}
\label{app:variance}
%%% ===================================

In this section, we provide an alternative proof of \cref{prop:variance_upper_bound} in the case of \ac{BW} circuits with periodic boundary conditions that relies on the $3$-design property of the multiqubits Clifford group.
In particular, it is based on the following result:
\begin{lemma}
	\label{lem:Clif_third_moment} 
	Let $n \in \NN$ and let $\gdim = 2^n$ be the dimension of a $n$-qubits system.
	If $0 \neq v \in \FF_2^{2n}$, then
	\begin{equation}
		\EE_{U \in \Cl_n(2)} U^{\otimes 3} \left( W(v)^{\otimes 2} \otimes W(u) \right) (U^\dagger)^{\otimes 3} = \frac{1}{\gdim^2 - 1} \delta_{u, 0} \left( \gdim \, \flip - \1_{\gdim^2} \right) \otimes \1_{\gdim} \quad \forall u \in \FF_2^{2n} \, .
	\end{equation}
	Moreover, 
	\begin{equation}
		\EE_{U \in \Cl_n(2)} U^{\otimes 3} \left( \1_{\gdim}^{\otimes 2} \otimes W(u) \right) (U^\dagger)^{\otimes 3} = \delta_{u,0} \1_{\gdim^3} \quad \forall u \in F_2^{2n} \, .
	\end{equation}
\end{lemma}
\begin{proof}
	First, we fix notations for the phase space representation of Clifford operators.
	
	Let $[\cdot, \cdot] : \FF_2^{2n} \times \FF_2^{2n} \rightarrow \FF_2$ be the standard symplectic product over $\FF_2^{2n}$.
	Let $\tilde \alpha_g: \FF_2^{2n} \rightarrow \ZZ_4$ be a centre fixing automorphism of the associated Heisenberg-Weyl group, where $g \in \Sp{2n}{2}$.
	One can prove that $\tilde \alpha_g = \alpha_{g} + [w, \cdot]$, where $\alpha_{g}: \FF_2^{2n} \rightarrow \FF_2$ is a suitable function satisfying the compatibility condition $\alpha_g(0) = 0$ \cite[Sec.~3.3]{heinrichStabiliserTechniquesTheir2021}.
	Notice also that $\abs{\Cl_n(2)} = d^2 \abs{\Sp{2n}{2}}$.
	
	With these notations, the action of $U \in \Cl_n(2)$ on Weyl operators can be written as
	\begin{equation}
		\label{eq:Cliff_definition_symplectic}
		U W(v) U^\dagger 
		\coloneqq
		\chi([a, v] + \alpha_g(v)) W( g(v) ) \, ,
	\end{equation}
	where $\chi(v) \coloneqq i^{- v_z \cdot v_x}$ denotes the character of $W(v)$.
	
	Hence,
	\begin{equation}
		\begin{aligned}
			\EE_{U \in \Cl_n(2)} 
			& 
			U^{\otimes 3} \left( W(v)^{\otimes 2} \otimes W(u) \right) (U^\dagger)^{\otimes 3} \\
			& = 
			\frac{1}{\abs{\Cl_n(2)}} \sum_{a \in \FF_2^{2n}} \chi([a, 2v + u]) \sum_{g \in \Sp{2n}{2}} \chi(2\alpha_g(v) + \alpha_g(u))
			 W(g(v))^{\otimes 2} \otimes W(g(u)) \\
			& = 
			\frac{1}{2^{2n}} \sum_{a \in \FF_2^{2n}} \chi([a, u]) \frac{1}{\abs{\Sp{2n}{2}}}
			\sum_{g \in \Sp{2n}{2}} \chi(\alpha_g(u)) W(g(v))^{\otimes 2} \otimes W(g(u)) \\
			& = 
			\frac{1}{\abs{\Sp{2n}{2}}} \, \delta_{u, 0} \sum_{g \in \Sp{2n}{2}} W(g(v))^{\otimes 2} \otimes W(g(u)) \\
			& = 
			\frac{1}{\abs{\Sp{2n}{2} \cdot v}} \, \delta_{u, 0} \sum_{w \in \Sp{2n}{2} \cdot v} W(w) \otimes W(w) \otimes \1_2 \, .
		\end{aligned}
	\end{equation}
	Going from the third to the fourth line, we used $\sum_{a \in \FF_2^{2n}} \chi([a, u]) = d^2 \delta_{u, 0}$.
	In the last step, we wrote the average over $\Sp{2n}{2}$ as an average over the orbit of $v$ under $\Sp{2n}{2}$.
	Notice that, since $v \neq 0$, $\Sp{2n}{2}$ acts transitively on $\FF_2^{2n} \setminus 0$ \cite{heinrichStabiliserTechniquesTheir2021}, the average over such orbit can be rewritten as an average over $\FF_2^{2n} \setminus 0$.
	Moreover, recalling that the flip operator has the following Pauli expansion:
	\begin{equation}
		\flip = \frac{1}{\gdim} \sum_{w \in \FF_2^{2n}} W(w) \otimes W(w) \, ,
	\end{equation}
	it holds that
	\begin{equation}
		\EE_{U \in \Cl_n(2)} U^{\otimes 3} \left( W(v)^{\otimes 2} \otimes W(u) \right) (U^\dagger)^{\otimes 3} 
		= \frac{1}{\gdim^2-1} \, \delta_{u, 0} \left( \gdim \flip - \1_{\gdim^2} \right) \otimes \1_{\gdim}.
	\end{equation}
	If $v = 0$, the assertion follows trivially from previous considerations.
\end{proof}
We will also need the following calculation:
\begin{lemma}
	\label{lem:partial_trace_Psym}
	Given two copies of a system of $3$ qubits, it holds
	\begin{equation}
		\Tr_3 \Psym[3] = \frac{d+2}{6} \left(\1 + \flip\right) \, .
	\end{equation}
\end{lemma}
\begin{proof}
	We prove the latter using tensor network diagrams.
	First, let us consider the decomposition $\Psym[3] = \frac{1}{6} \left( \1 + P_{(1,2)} + P_{(1, 3)} + P_{(2, 3)} + P_{(1, 2, 3)} + P_{(1, 3, 2)} \right)$, $P_{(\cdot)}$ are unitary operators associated with elements of the permutation group $\Sym_3$, and each $a \in \Sym_3$ is represented in cyclic notation.
	Then, 
	\begin{equation*}
		\begin{tikzpicture}
			\node[inner sep=0pt] (full)
			{\includegraphics{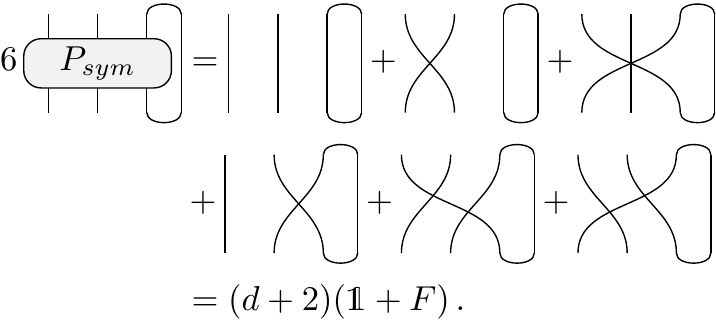}};	
		\end{tikzpicture}
	\end{equation*}
\end{proof}

\begin{proposition}
	\label{prop:variance_upper_bound_appendix}
	For any state $\rho$, estimate $W(v)$ using \ac{BW} shadows with periodic boundary conditions.
	Then, the variance of the estimator depends only on $v \in \FF_2^{2n}$, and
	\begin{equation}
		\varBW(v)
		\leq
		\frac{1}{\osandwich v S v} \, .
	\end{equation}
\end{proposition}
\begin{proof}
	In the following, we will denote by $D$ the cyclic shift operator between Hilbert spaces as before, so that a random brickwork unitary is given by $U = D U_2 D^\ad U_1$, where $U_{i}, \ i = 1,2$, is the tensor product of two-local Haar random unitaries.
	For a given operator $A \in \L(\CC^{2} \otimes \CC^{2})$, we will also consider the operator acting on three copies of two qubit sites $A_{(3)} \coloneqq A \otimes A \otimes A \in \L(\CC^{8} \otimes \CC^{8})$.
	
	According to the shadow estimation protocol, we estimate the expectation value $\Tr[W(v) \rho]$ of some Pauli observable $W(v)$ by measuring $\rho$ many times in the computational basis after having applied $U \sim \mu$, where $\mu$ is a probability measure on the ensemble of \ac{BW} operators. 
	Then, a single such sample has a variance $\varBW(v, \rho)$ bounded as
	\begin{equation}
		\begin{aligned}
			\varBW(v, \rho) &\coloneqq 
			\sum_i \EE_{U\sim \mu} \obraket{W(v)}{\tilde E_{i,U}}^2 \obraket{E_{i,U}}{\rho} - \Tr[W(v)\rho]^2
			\\
			&\leq
			\sum_i \EE_{U\sim \mu} \Tr\myleft[S^{-1}(W(v)) U^\ad E_i U\myright]^2 \, \Tr\myleft[U^\dagger E_{i} U\rho\myright]
			\\
			&=
			\gdim \, \EE_{U\sim \mu} \Tr\myleft[ U^{\otimes 3\ad} \ketbra 00^{\otimes 3} U^{\otimes 3} \left(S^{-1}(W(v))^{\otimes 2} \otimes \rho\right) \myright] 
			\\
			&= 
			\frac{\gdim}{\osandwich {v}{S}{v}^2 } \, \EE_{U\sim \mu} 
			\Tr\myleft[ U^{\otimes 3\ad} \ketbra 00^{\otimes 3} U^{\otimes 3} \left(W(v)^{\otimes 2} \otimes \rho\right) \myright] 
			\\
			&=
			\frac{1}{(2 \sqrt{5})^n}
			\frac{\gdim}{\osandwich v S v^2 } \, 
			\EE_{U_2} 
			\Tr\myleft[ \Psym[3]^{\otimes n/2} D_{(3)} \, U_2^{\otimes 3} \left( D_{(3)}^{\ad} W(v)^{\otimes 2} \otimes \rho \, D_{(3)}\right) U_2^{\otimes 3\ad} D_{(3)}^{\ad}  \myright]  \, ,
		\end{aligned}
	\end{equation}  
	where in the last step we applied again \cref{eq:k_moment_zero} from \cref{lem:Haar_second_moment} and $U = D U_2 D^\ad U_1$.
	Notice that $\Psym[3]$ acts on triples of two neighboring qubit sites.
	
	Consider now the expansion in the Pauli basis $\rho = \sum_{u \in \FF_2^{2n}} c_u W(u)$, where $c_u \coloneqq \frac{1}{d} \Tr(W(u)\rho)$, and for any $w \in \FF_2^{2n}$ consider the decomposition
	\begin{align}
		D^\ad W(w) D &= W(w_{2,3}) \otimes \dots \otimes W(w_{n,1}) \, ,
	\end{align}
	which agrees with the structure of the second layer of the \ac{BW} circuit.
	Then, by \cref{lem:Clif_third_moment}, we have
	\begin{equation}
		\EE_{U_2} 
		U_2^{\otimes 3} \left( D_{(3)}^{\ad} W(v)^{\otimes 2} \otimes W(u) \, D_{(3)}\right) U_2^{\otimes 3\ad}
		=
		\delta_{u,0} \bigotimes_{i \in [n/2]} Q_{\hat v_i} \otimes \1_{\gdim} \, ,
	\end{equation}
	from which it follows
	\begin{equation}
		\begin{aligned}
			\sigma^2(v, \rho)
			& \leq
			\frac{1}{(2 \sqrt{5})^n} \frac{\gdim}{\osandwich v S v^2 } \, c_0 \, 
			\Tr\myleft[ \Psym[3]^{\otimes n/2} \ 
			D_{(3)} \bigotimes_{i \in [n/2]} Q_{\tilde v_i} \, D_{(3)}^{-1} \otimes \1_{\gdim} \myright] \, ,
		\end{aligned} 
	\end{equation}
	where $Q_{\tilde v_i}$ is defined in \cref{eq:Q_v}.
	Hence, $\sigma^2(v, \rho) \equiv \sigma^2(v)$, since $c_0 = 1/\gdim$.
	According to \cref{lem:Clif_third_moment}, this means that each Clifford unitary in the second layer depolarizes any dependency from the corresponding two-qubits Weyl operator appearing in the decomposition of $\rho$ in the Pauli basis; periodic boundary conditions ensure that this applies to each pair of qubits.
	Finally, by \cref{lem:partial_trace_Psym}, 
	\begin{equation}
		\begin{aligned}
			\sigma^2(v)
			& \leq
			\frac{1}{(2 \sqrt{5})^n}
			\frac{1}{\osandwich v S v^2 } \, 
			\Tr\myleft[ \Tr_3(\Psym[3]^{\otimes n/2}) \ 
			D_{(2)} \bigotimes_{i \in [n/2]} Q_{\tilde v_i} \, D_{(2)}^{-1} \myright]
			\\
			& =
			\frac{1}{(2 \sqrt{5})^{2n}}
			\frac{1}{\osandwich v S v^2 } \, 
			\Tr\myleft[ (\1 + \flip)^{\otimes n/2} \ 
			D_{(2)} \bigotimes_{i \in [n/2]} Q_{\tilde v_i} \, D_{(2)}^{-1} \myright]
			\\
			& =
			\frac{1}{\osandwich v S v} \, ,
		\end{aligned}
	\end{equation}	
\end{proof}
Finally, notice that \cref{lem:Clif_third_moment} does not hold for arbitrary values of the local dimension.
Indeed, in odd dimensions, the flip operator has the Weyl expansion
\begin{equation}
	\flip = \frac{1}{\gdim} \sum_{v \in \FF_\ldim^2} W(v) \otimes W(-v) \, ,
\end{equation}
meaning the operator $\frac{1}{\gdim} \sum_{v \in \FF_2^{2n}} W(v) \otimes W(v)$ admits a nice expression for fields of characteristic $2$ only, and the proof of \cref{prop:variance_upper_bound_appendix} holds for qubit systems only.

%%% ===================================
\section{Numeric bounds on the variance}
\label{app:bounds}
%%% ===================================

For the construction of the classical shadow, the exact expressions stated in Theorem~\ref{thm:main} are required.
We here derive simpler (and looser) bounds for controlling the variance.

Let us start with the case relevant for $\partBW(v) = (n/2)$.
Set $a \coloneqq (\sqrt{41} + 5)^{1/2} / (5\sqrt2)$ and $b \coloneqq \i  (\sqrt{41} - 5)^{1/2} / (5\sqrt2)$. 
We have $2 < 1 / a < 2.1$, $\frac1{1 - |b/a|^2} \leq 1.2$ and $\frac1{1 + |b/a|^2} \geq .89$.
Further, since $|b/a| < 1$ and assuming $n \geq 2$ and $n$ even, 
\begin{equation}
	\frac1{\frameopPeriodic(n)} 
	= \frac1{a^n + b^n} = \frac1{a^n} \frac1{1 + (b/a)^n} \leq \frac1{a^n} \frac1{1 - |b/a|^2} \leq 1.2 \cdot 2.1^n \, ,
\end{equation}
and analogously
\begin{equation}
	\frac1{\frameopPeriodic(n)} \geq \frac1{a^n} \frac1{1 + |b/a|^2} \geq 0.8 \cdot 2^n\,. 
\end{equation}

In the same way, we can bound $1/\frameopOpen(n)$.
To this end, further set
$c = 5 (25 - 3 \sqrt{41}) / (2 \sqrt{41})$ 
and $d = 5 (25 + 3 \sqrt{41}) / (2 \sqrt{41})$.
We have, up to adjusting the phase of $b$,
\begin{equation}\label{eq:secbound}
	\frac1{\frameopOpen(n)} = \frac1{c a^n + d b^n} = \frac1{ca^n} \frac1{1 + (d/c) (b/a)^n}\,
\end{equation}
and, thus, for $n \geq 4$ we have
\begin{equation}
	0.3 \cdot 2^{n} < \frac1{ca^n} \frac1{1 + (d/c) |b/a|^4} \leq \frac1{\frameopOpen(n)} \leq \frac1{ca^n} \frac1{1 - (d/c) |b/a|^4} < 0.6 \cdot 2.1^n.
\end{equation}
Note that \cref{eq:secbound} goes to $c^{-1} a^{-n} \approx 0.44 \cdot 2.1^n$ for large $n$ as the second fraction becomes $1$ asymptotically. 
Similarly, $\frac1{\frameopPeriodic(n)}$ asymptotically becomes $a^{-n}$.
The deviation from this asymptotic scaling is small already for small $n$. 
E.g., the relative error of the asymptotic approximation is smaller than $10^{-2}$ for $n \geq 6$.
Asymptotically the frame operator elements for periodic and open boundary conditions only differ by a constant factor $c^{-1} \approx 0.44$.

Setting $\Gamma \equiv \frac{1}{c} \frac{1}{1 - (d/c) \abs{b/a}^4}$ and $\Delta \equiv \frac{1}{a}$, 
a bound of the form $1 / \frameopOpen(n) \leq \Gamma \Delta^n$, implies that the variance is dominated by
\begin{equation}
	\varBW \leq \prod_{l\in \partBW(v)} \frameopOpen(2l + 2)^{-1}
	\leq \prod_{l\in \partBW(v)} (\Gamma \Delta^2) \Delta^{2l}
	\leq (\Gamma \Delta^2)^{|\partBW(v)|} \Delta^{2\Sigma(\partBW(v))},\,
\end{equation}
where $|\partBW(v)|$ denotes the length of the tuple $\partBW(v)$, i.e.\ the number of parts in the partition, and 
$\Sigma(\partBW(v)) \coloneqq \sum_{l \in |\partBW(v)|} l = |\tilde v|$ is the cumulative length of all parts.

Inserting the previous bounds for $\Gamma$ and $\Delta$ (without intermediate rounding), we conclude that 
\begin{equation}\label{eq:generalVarianceBound}
	\varBW \leq 2.2^{|\partBW(v)|} 4.4^{\Sigma(\partBW(v))}\,.
\end{equation}

When comparing to the variance of \ac{LC} circuits in the case where $\partBW = (n/2)$ or $(n/2 - 1)$, we are interested in ensuring that $\Gamma \Delta^n < 3^{\abs{\suppLC(v)}}$.
This is the case when 
\begin{equation}
	\abs{\suppLC(v)} \geq n \log_3\Delta + \log_3\Gamma \, ,
\end{equation}
which for open and periodic boundary conditions translates to the sufficient condition
\begin{equation}
	\abs{\suppLC(v)} > 0.68 n + 0.12 \, 
\end{equation} 
(the constant term for open boundary conditions is actually negative).

More generally, \cref{eq:generalVarianceBound} is smaller than $3^{\abs{\suppLC(v)}}$ if
\begin{equation}
	\abs{\suppLC(v)} \geq 0.8 |\partBW(v)|  + 1.4 \Sigma(\partBW(v))\,.
\end{equation}

%%% ===================================
\section{More details on numerical experiments}
\label{app:numerics}
%%% ===================================

In this section, we describe in detail the procedure used for our numerical experiments, which are implemented in the following repository: \url{https://github.com/MirkoArienzo/shadow_short_circuits}.
First, given $\rho = \ketbra{0}{0}$ and a Pauli observable $W(v)$, \cref{eq:reconstruction_shadow} becomes
\begin{equation}
	\label{eq:pauli_task_zero_state}
	\sandwich{0}{W(v)}{0} 
	=
	\frac{1}{\osandwich{v}{\frameop}{v}}
	\sum_i \EE_{U \sim \mu} \sandwich{i}{U W(v) U^\ad}{i} \abs{ \sandwich{i}{U}{0} }^2 \, .
\end{equation}
As discussed above, $U$ is chosen to be a Clifford operator, which is represented by a pair $(g, a)$, with $g \in \Sp{2n}{2}$, and $a \in \FF_2^{2n}$.
Then, writing $U = \bigotimes_{i=1}^{n/2} \bigotimes_{j=1}^{n/2} U_1^{(i)} U_2^{(j)}$, each local symplectic matrix is sampled using K\"onig-Smolin's algorithm \cite{koenig_how_2014}, and $a$ is a uniformly distributed vector in $\FF_2^{2n}$.
Then, samples $\{ (U_j, i_j) \}_{j=1}^m$ are drawn according to standard stabilizer simulation techniques \cite{gottesmanHeisenbergRepresentationQuantum1998, aaronsonImprovedSimulationStabilizer2004}, and the estimator is given by the following empirical average:
\begin{equation}
	\hat w (v)
	=
	\frac{1}{m \osandwich{v}{\frameop}{v}} 
	\sum_{j=1}^m \sandwich{i_j}{U_j W(v) U_j^\ad}{i_j} \, .
\end{equation}

A single estimate requires the calculation of the phase function appearing in \cref{eq:Cliff_definition_symplectic}, which can be done in time $\LandauO(n^3)$ \cite{heinrichStabiliserTechniquesTheir2021}.
However, when the observable is of $Z$-type, we can avoid this calculation, and speed up the simulation.
To prove this fact, let us consider the decomposition $\FF_2^{2n} = Z_n \oplus X_n$, and label the computational basis by binary vectors $i\in\FF_2^n$.
Then,
\begin{equation}
	\begin{aligned}
		\sandwich{i}{U W(v) U^\ad}{i}
		& =
		(-1)^{\alpha_g(v) + [a, gv] + (gv)_z \cdot i} 1_{Z_n}(gv) \, ,
	\end{aligned}
\end{equation}
where $(gv)_z\in Z_n\simeq\FF_2^n$ is the $Z$ part of the vector $gv\in\FF_2^{2n}$, and $1_{Z_n}$ is the indicator function on $Z_n$.
Then, suppose the outcome of the latter is non-zero, so $(gv)_x = 0$.
Hence, since $v_x = 0$ by assumption, we find:
\begin{equation}
	\begin{aligned}
		[a, gv] + (gv)_z \cdot i
		& = 
		(gv)_z \cdot i_0
		= 
		(gv)_z \cdot \sum_j i_j (ge_j)_x \\
		& =
		\sum_j i_j [gv, ge_j] 
		= 
		\sum_j i_j [v, e_j] \\
		& = 
		0 \, ,
	\end{aligned}
\end{equation}
where we wrote $i = i_0 + a_x$ for some $i_0 \in \FF_2^{n}$, and then we considered the decomposition $i_0 = \sum_j i_j (ge_j)_z$, where $\{e_j\} \subset Z_n$ is the canonical basis.

Then, from \cref{eq:pauli_task_zero_state} we get
\begin{equation}
	\begin{aligned}
		\label{eq:pauli_task_zero_state_intermediate}
		\osandwich{v}{\frameop}{v}
		& =
		\sum_i \EE_{U \sim \mu} \sandwich{i}{U W(v) U^\ad}{i} \abs{ \sandwich{i}{U}{0} }^2 \\
		& =
		\sum_i \EE_{U \sim \mu} (-1)^{\alpha_g(v)} \abs{ \sandwich{i}{U}{0} }^2 \, .
	\end{aligned}
\end{equation}
Define now
\begin{equation}
	p_{\pm} 
	\coloneqq
	\frac{1}{\abs{\Sp{2n}{2}}} \abs{ \{ g \in \Sp{2n}{2} \mid gv \in Z_n \, , \ (-1)^{\alpha_g(v)} = \pm 1 \} } \, .
\end{equation}
Then,
\begin{equation}
	\begin{aligned} 
		p_+ + p_-
		& = 
		\frac{1}{\abs{\Sp{2n}{2}}}
		\abs{ \{ g \in \Sp{2n}{2} \mid gv \in Z_n \} } \\ 
		& = 
		\frac{1}{\abs{\Sp{2n}{2}}} \abs{\stab(v) } \cdot \abs{ Z_n \setminus 0 } \\ 
		& =
		\frac{1}{\abs{\Sp{2n}{2}}} \frac{\abs{\Sp{2n}{2}}}{\abs{\Sp{2n}{2} \cdot v}} \cdot \abs{ Z_n \setminus 0 } \\
		& =
		\frac{2^n-1}{2^{2n}-1} = \frac{1}{2^n+1} \, ,
	\end{aligned}
\end{equation}
where $\stab(v)$ denotes the set of stabilizers of $v$.

On the other hand, $p_+$ and $p_-$ also have the interpretation of frequencies of $\pm 1$ outcomes in \cref{eq:pauli_task_zero_state_intermediate} respectively.
This means
\begin{equation}
	\osandwich{v}{\frameop}{v} 
	= 
	p_+ - p_- 
	=
	\frac{\abs{Z_n}}{2^{2n}-1}
	=
	\frac{1}{2^n+1}
	\, ,
\end{equation}
from which it follows $p_- = 0$.

In conclusion, whenever $v_x=0$, we only need to check if $(gv)_x$ is trivial.

%%%=============================================

\begin{acronym}[POVM]\itemsep.5\baselineskip
\acro{AGF}{average gate fidelity}

\acro{BOG}{binned outcome generation}

\acro{BW}{brickwork}

\acro{CP}{completely positive}
\acro{CPT}{completely positive and trace preserving}
\acro{CS}{compressed sensing} 

\acro{DFE}{direct fidelity estimation} 
\acro{DM}{dark matter}

\acro{GST}{gate set tomography}
\acro{GUE}{Gaussian unitary ensemble}

\acro{HOG}{heavy outcome generation}

\acro{LC}{local Clifford}

\acro{MBL}{many-body localization}
\acro{ML}{machine learning}
\acro{MLE}{maximum likelihood estimation}
\acro{MPO}{matrix product operator}
\acro{MPS}{matrix product state}
\acro{MUBs}{mutually unbiased bases} 
\acro{MW}{micro wave}

\acro{NISQ}{noisy and intermediate scale quantum}

\acro{POVM}{positive operator valued measure}
\acro{PVM}{projector-valued measure}

\acro{QAOA}{quantum approximate optimization algorithm}
\acro{QML}{quantum machine learning}
\acro{QMT}{measurement tomography}
\acro{QPT}{quantum process tomography}

\acro{RDM}{reduced density matrix}

\acro{SFE}{shadow fidelity estimation}
\acro{SIC}{symmetric, informationally complete}
\acro{SPAM}{state preparation and measurement}

\acro{RB}{randomized benchmarking}
\acro{rf}{radio frequency}

\acro{TT}{tensor train}
\acro{TV}{total variation}

\acro{VQA}{variational quantum algorithm}

\acro{VQE}{variational quantum eigensolver}

\acro{XEB}{cross-entropy benchmarking}

\end{acronym}

\addcontentsline{toc}{chapter}{References}
\bibliographystyle{myapsrev4-2}
\bibliography{mk,shadows}

%\printbibliography
%%%=============================================
\end{document}